\documentclass[11pt,a4paper]{amsart}
\usepackage[margin=25mm]{geometry}
\usepackage[numbers]{natbib}
\usepackage{hyperref}

\usepackage{amssymb,amsmath}
\usepackage{amsthm}
\usepackage{bbm, bm}
\usepackage{stmaryrd}
\usepackage{color}
\usepackage{graphicx}
\usepackage{caption}
\usepackage{float}
\usepackage{wrapfig}
\usepackage{subcaption}
\usepackage{amsaddr}
\usepackage{multicol}
\usepackage{enumerate}
\usepackage{lineno}

\usepackage{abstract}
\usepackage{orcidlink}

\input xy
\xyoption{all} 

\numberwithin{equation}{section}

\newtheorem{theorem}{Theorem}[section]
\newtheorem{corollary}{Corollary}[section]
\newtheorem{lemma}{Lemma}[section]
\newtheorem{proposition}{Proposition}[section]

\newtheorem{observation}{Observation}[section]
\theoremstyle{definition}
\newtheorem{definition}{Definition}[section]
\newtheorem{example}{Example}[section]

\newenvironment{warning}[1][Warning.]{\begin{trivlist}
\item[\hskip \labelsep {\bfseries #1}]}{\end{trivlist}}

\newenvironment{remark}[1][Remark.]{\begin{trivlist}
\item[\hskip \labelsep {\bfseries #1}]  }{ \end{trivlist}}

\newcommand{\Dleft}{[\hspace{-1.5pt}[}
\newcommand{\Dright}{]\hspace{-1.5pt}]}
\newcommand{\SN}[1]{\Dleft #1 \Dright}

\newcommand{\p}{\mbox{\boldmath$\rho$}}

\newcommand{\InHom}{\mbox{$\underline{\Hom}$}}
\newcommand{\InLin}{\mbox{$\underline{\Lin}$}}

\newcommand{\rmd}{\textnormal{d}}

\DeclareMathOperator{\Vect}{Vect}

\DeclareMathOperator{\Span}{Span}

\DeclareMathOperator{\Der}{Der}

\DeclareMathOperator{\Hom}{Hom}

\DeclareMathOperator{\Lin}{Lin}

\newcommand{\catname}[1]{\textnormal{\texttt{#1}}}

\font\black=cmbx10 \font\sblack=cmbx7 \font\ssblack=cmbx5 \font\blackital=cmmib10  \skewchar\blackital='177
\font\sblackital=cmmib7 \skewchar\sblackital='177 \font\ssblackital=cmmib5 \skewchar\ssblackital='177
\font\sanss=cmss10 \font\ssanss=cmss8 
\font\sssanss=cmss8 scaled 600 \font\blackboard=msbm10 \font\sblackboard=msbm7 \font\ssblackboard=msbm5
\font\caligr=eusm10 \font\scaligr=eusm7 \font\sscaligr=eusm5  \font\fraktur=eufm10
\font\sfraktur=eufm7 \font\ssfraktur=eufm5 
\font\bsymb=cmsy10 scaled\magstep2
\def\all#1{\setbox0=\hbox{\lower1.5pt\hbox{\bsymb
       \char"38}}\setbox1=\hbox{$_{#1}$} \box0\lower2pt\box1\;}
\def\exi#1{\setbox0=\hbox{\lower1.5pt\hbox{\bsymb \char"39}}
       \setbox1=\hbox{$_{#1}$} \box0\lower2pt\box1\;}

\def\tx#1{{\fam0\relax#1}}

\newfam\bifam
\textfont\bifam=\blackital \scriptfont\bifam=\sblackital \scriptscriptfont\bifam=\ssblackital

\newfam\blfam
\textfont\blfam=\black \scriptfont\blfam=\sblack \scriptscriptfont\blfam=\ssblack

\newfam\bbfam
\textfont\bbfam=\blackboard \scriptfont\bbfam=\sblackboard \scriptscriptfont\bbfam=\ssblackboard

\newfam\ssfam
\textfont\ssfam=\sanss \scriptfont\ssfam=\ssanss \scriptscriptfont\ssfam=\sssanss
\def\sss#1{{\fam\ssfam\relax#1}}

\newfam\clfam
\textfont\clfam=\caligr \scriptfont\clfam=\scaligr \scriptscriptfont\clfam=\sscaligr

\newfam\frfam
\textfont\frfam=\fraktur \scriptfont\frfam=\sfraktur \scriptscriptfont\frfam=\ssfraktur

\def\hpb#1{\setbox0=\hbox{${#1}$}
    \copy0 \kern-\wd0 \kern.2pt \box0}
\def\vpb#1{\setbox0=\hbox{${#1}$}
    \copy0 \kern-\wd0 \raise.08pt \box0}

\def\pmb#1{\setbox0\hbox{${#1}$} \copy0 \kern-\wd0 \kern.2pt \box0}
\def\pmbb#1{\setbox0\hbox{${#1}$} \copy0 \kern-\wd0
      \kern.2pt \copy0 \kern-\wd0 \kern.2pt \box0}
\def\pmbbb#1{\setbox0\hbox{${#1}$} \copy0 \kern-\wd0
      \kern.2pt \copy0 \kern-\wd0 \kern.2pt
    \copy0 \kern-\wd0 \kern.2pt \box0}
\def\pmxb#1{\setbox0\hbox{${#1}$} \copy0 \kern-\wd0
      \kern.2pt \copy0 \kern-\wd0 \kern.2pt
      \copy0 \kern-\wd0 \kern.2pt \copy0 \kern-\wd0 \kern.2pt \box0}
\def\pmxbb#1{\setbox0\hbox{${#1}$} \copy0 \kern-\wd0 \kern.2pt
      \copy0 \kern-\wd0 \kern.2pt
      \copy0 \kern-\wd0 \kern.2pt \copy0 \kern-\wd0 \kern.2pt
      \copy0 \kern-\wd0 \kern.2pt \box0}

\mathchardef\za="710B  
\mathchardef\zb="710C  
\mathchardef\zg="710D  
\mathchardef\zd="710E  
\mathchardef\zve="710F 
\mathchardef\zz="7110  
\mathchardef\zh="7111  
\mathchardef\zvy="7112 
\mathchardef\zi="7113  
\mathchardef\zk="7114  
\mathchardef\zl="7115  
\mathchardef\zm="7116  
\mathchardef\zn="7117  
\mathchardef\zx="7118  
\mathchardef\zp="7119  
\mathchardef\zr="711A  
\mathchardef\zs="711B  
\mathchardef\zt="711C  
\mathchardef\zu="711D  
\mathchardef\zvf="711E 
\mathchardef\zq="711F  
\mathchardef\zc="7120  
\mathchardef\zw="7121  
\mathchardef\ze="7122  
\mathchardef\zy="7123  
\mathchardef\zf="7124  
\mathchardef\zvr="7125 
\mathchardef\zvs="7126 
\mathchardef\zf="7127  
\mathchardef\zG="7000  
\mathchardef\zD="7001  
\mathchardef\zY="7002  
\mathchardef\zL="7003  
\mathchardef\zX="7004  
\mathchardef\zP="7005  
\mathchardef\zS="7006  
\mathchardef\zU="7007  
\mathchardef\zF="7008  
\mathchardef\zW="700A  
\mathchardef\zC="7009  

\newcommand{\be}{\begin{equation}}
\newcommand{\ee}{\end{equation}}

\newcommand{\bea}{\begin{eqnarray}}
\newcommand{\eea}{\end{eqnarray}}
\def\*{{\textstyle *}}
\newcommand{\R}{{\mathbb R}}

\newcommand{\Z}{{\mathbb Z}}

\newcommand{\s}{{\textstyle *}}

\def\Hom{\sss{Hom}}

\def\ul{\underline}

\def\Vect{\sss{Vect}}
\def\Lin{\sss{Lin}}

\def\sT{{\sss T}}

\def\xi{\tx{i}}

\def\s*{{\scriptstyle *}}

\def\cO{\mathcal{O}}

\def\ul{\underline}

\newcommand{\beas}{\begin{eqnarray*}}
\newcommand{\eeas}{\end{eqnarray*}}

\title{Derived Associative Algebras and Cyclic Cohomology of Q-manifolds} 
\author{Andrew James Bruce \,\orcidlink{0000-0001-8197-2263}
 }  
\address{Independent Researcher, Cardiff, United Kingdom }
   \email{andrewjamesbruce@googlemail.com}

   \date{\today}
 
\begin{document}
 \maketitle
\vspace{-20pt}
\begin{abstract}{\noindent Loday's derived product is revisited in the setting of Q-manifolds, i.e., supermanifolds equipped with an odd vector field that `squares to zero'. We interpret the Grassmann odd product as a standard product upon shifting the grading, which implies the existence of a derived (noncommutative) associative $\mathbb{Z}_2$-graded algebra associated with any Q-manifold. We apply Connes' cyclic cohomology to the derived associative algebra, giving a new cohomology on Q-manifolds distinct from the standard cohomology: we refer to this as the derived cyclic cohomology of a Q-manifold. The derived cyclic cocycles are interpreted as classically BRST-invariant functionals within the BV--BFV--BRST formalism or generalised Ruelle--Sullivan currents when applied to regular foliations via their foliation Lie algebroids.   }\\
\noindent {\Small \textbf{Keywords:}~Q-manifolds;~Derived Associative Algebras;~Derived Cyclic Cohomology;~BV--BRST Mechanics}
\end{abstract}
\tableofcontents
\emph{
I'm very good at integral and differential calculus,\\
I know the scientific names of beings animalculous:\\
In short, in matters vegetable, animal, and mineral,\\
I am the very model of a modern Major-General.}
\smallskip

\begin{flushright}
W.S. Gilbert \& A.S. Sullivan \emph{The Pirates of Penzance}, 1879
\end{flushright}

\section{Introduction and Background}
\subsection{Introduction}
Q-manifolds, that is, supermanifolds equipped with an integrable odd vector field,  are found in many areas of mathematical physics and differential geometry. Specifically, Q-manifolds arise in the theory of Lie algebroids, Poisson geometry, Courant algebroids, $L_\infty$-algebroids, and related settings. Even  the de Rham complex and Cartan calculus have a very natural formulation in terms of a Q-manifold. From the perspective of this paper, Q-manifolds (or their infinite-dimensional generalisations) appear in physics via the BV--BFV--BRST\footnote{\textbf{B}atalin--\textbf{V}ilkovisky, \textbf{B}atalin--\textbf{F}radkin--\textbf{V}ilkovisky, and  \textbf{B}ecchi--\textbf{R}ouet--\textbf{S}tora--\textbf{T}yutin.} formalism, see \cite{Alexandrov:1997,Schwarz:1993}. The philosophy we adopt is that Q-manifolds provide a simplified environment to uncover mathematical aspects of the BV--BRST formalism in field theory. More correctly, we view Q-manifolds through the lens of BRST gauge mechanics by interpreting the homological vector field as the BRST differential; see \cite{Henneaux:1992}, for example.  By restricting attention to finite-dimensional supermanifolds, we avoid the technical issues with field theory.\par 
It must be stressed that BRST symmetry, and thus BRST-cohomology, is vital in the consistent quantisation of gauge theories, general relativity (as an effective theory),  and string theory. In particular, local BRST-cohomology and the descent equations are the foundation of the modern understanding of anomalies; see Bertlmann \cite{Bertlmann:2000} for a review. Q-manifolds are vital in the AKSZ construction of sigma models, see \cite{Alexandrov:1997}; and so fundamental in Kontsevich's work on deformation quantisation of Poisson manifolds, see \cite{Kontsevich:2003}. We also remark that Q-manifolds provide concrete, finite-dimensional (local) models for derived manifolds within derived differential geometry, see \cite{Carchedi:2023}.  In short, Q-manifolds are ubiquitous across geometric approaches to physics and modern differential geometry.  \par 
In this paper, we revisit Loday's derived product/multiplication on a differential graded algebra in the geometric setting of Q-manifolds. Loday's motivation (see \cite{Loday:2001}) was the construction of associative dialgebras in order to build Leibniz--Loday algebras, i.e., `non-skew-symmetric Lie algebras'.  The core idea of this paper is to take the construction of a derived multiplication as a stand-alone notion. The subtlety is that any such derived product $f \star g = \pm Q(f)g$, where $Q$ is a differential (homological vector field),  carries non-zero Grassmann parity. In short, we have an associative odd multiplication. That is, the Grassmann parity of the product of two functions is the sum of Grassmann parities of the functions, plus one. Another immediate complication is that the lack of a unit element for the derived product.   \par 
Odd multiplication is already an exotic notion and seems discordant with standard graded algebra.  We will show that this can be largely remedied by employing a shifted grading which renders the derived product a standard $\Z_2$-graded associative algebra. We stress that this shifted grading is formal in the sense that it is just a relabelling of even and odd functions. We are not constructing a new supermanifold using this shift; we have a bookkeeping device to track the sign factors consistently.  Thus, by employing a shifted grading, the global functions on a Q-manifold come equipped with a \emph{derived associative algebra}, which is a standard associative superalgebra with respect to the shifted grading.  We further stress that the derived associative algebra associated with a Q-manifold is not (super)commutative. Thus, the core ethos of this work is the following: \par 
\smallskip
\noindent \emph{What happens if we treat the derived product on a Q-manifold with the machinery of noncommutative geometry\footnote{See, for example \cite{Connes:1994,Gracia-Bondía:2001,Madore:1999,Manin:1991} for introductions to noncommutative geometry.}? } \par  
\smallskip
We remark that odd forms of multiplication can be found in the setting of Stasheff's $A_\infty$-algebras  (strongly homotopy associative algebras) via suspension/shift. That is, the even binary operation $m_2 : A \otimes A \rightarrow A$ can be shifted to an odd binary operation $b_2: A[1] \otimes A[1] \rightarrow A[1]$. This may be understood in terms of the operadic suspension,  i.e., the operad whose operations have been shifted in homological degree and the appropriate sign factors inserted.  For an introduction to $A_\infty$-algebras, the reader may consult Keller \cite{Keller:2001}. The bar-cobar construction linking dg-algebras and dg-coalgebras similarly requires a suspension/shift. For an introduction to the bar-cobar construction, the reader may consult \cite[Chapter 2]{Loday:2012}.\par    
 We apply Connes' (topological) cyclic cohomology, suitably `superised' following Kastler \cite{Kastler:1988}, to the derived associative algebra associated with any Q-manifold.  In doing so, we obtain a cohomology theory quite independent of the standard cohomology of a Q-manifold: this cohomology theory we refer to as the \emph{derived cyclic cohomology} of a Q-manifold. While the calculation of cyclic cohomology groups for general noncommutative algebras is notoriously difficult, we present some simple examples of derived cyclic zeroth-cohomology groups. The zeroth cohomology is interpreted within the finite-dimensional BV--BRST formalism as equivalence classes of classically BRST-invariant functionals  $\tau : C^\infty(M) \rightarrow \R$. Mathematically, these functionals are Q-invariant in the sense that $\tau \circ Q  =0$. Similarly, the higher cohomology groups are interpreted as equivalence classes of BRST-invariant derived cyclic multi-functionals. We apply this machinery to regular foliations via their foliation Lie algebroid and show that the derived cyclic $0$-cocycles include Ruelle--Sullivan currents (see \cite{Ruelle:1975}). Moreover, we construct derived cyclic $n$-cocycles associated with a regular foliation. \par
\begin{warning}
We emphasise that the descriptor `derived' used throughout this work indicates derived in the sense of Loday, i.e., the structures are built from a standard one and a differential.  It should not be confused with the related notion found within the complementary framework of derived algebraic geometry (see \cite{Gaitsgory:2017} for an introduction).
\end{warning}
\medskip

\noindent \textbf{Key Results:} Amongst other results, we establish the following.
\begin{itemize}
\item \textbf{Theorem \ref{thm:DerCycCoFunct}} states that the derived cyclic cohomology is a covariant functor
$$HC^n_\lambda :\catname{QMan} \rightarrow \catname{SVec}\,,$$
 ($n=0,1, 2, \cdots$) from the category of Q-manifolds to the category of super vector spaces.
\item \textbf{Proposition \ref{prop:CochainsInv}} states that all $b$-closed cochains $\phi^n$  are Q-invariant, or  in physics terminology, BRST-invariant, in the sense that  
$$\phi^n \circ (Q \otimes \cdots \otimes Q) =0\,.$$
\item  \textbf{Theorem \ref{thm:DerCycCoOddLine}} gives the derived cyclic cohomology of the odd line equipped with a non-singular homological vector field 
$$
HC^n_\lambda(\R^{0|1}, Q) = \begin{cases} 
\R^{0|1} = 0 \oplus \R \,, & \text{if~} n \text{~is even} \\ 
0 = 0\oplus 0 \,, & \text{if~} n \text{~is odd} 
\end{cases}\,.
$$
\item \textbf{Proposition \ref{prop:LieAlgH0}} states that for  (pure even) $n$-dimensional Lie algebras 
$$HC^0_\lambda(\Pi \mathfrak{g}, \rmd_{CE}) \cong \mathcal{Z}_\bullet(\mathfrak{g})= \bigoplus_{k=0}^n \mathcal{Z}_k(\mathfrak{g})\,,$$
where $\mathcal{Z}_k(\mathfrak{g})$ are the Lie algebra cycle groups.
\item \textbf{Theorem \ref{the:ObsDeform}} states that the infinitesimal deformation problem form  derived cyclic $0$-cocycles is controlled by the derived cyclic cohomology.  More carefully, if $X \in \Vect(M)$ is such that $\mathcal{L}_Q X =0$, then the obstruction to the infinitesimal deformation of $\tau$ is given by the derived cyclic class of $\phi(f_0, f_1):= \tau(\pm X(f_0 f_1))$, which we refer to as the deformation class of $\tau$ along $X$. 
\end{itemize} 

\medskip
\noindent \textbf{Organisation:} We continue this introductory section with Subsection \ref{subsec:Qman} in which the fundamental aspects of the theory of Q-manifolds are recalled. In Section \ref{sec:DerAlg} we define the derived product and explore its properties. In Section \ref{sec:DerCycCoho} we move on to the derived cyclic cohomology. We end with some concluding remarks and open questions in Section \ref{sec:ConRem}.
%
%
\subsection{Recollection of Q-manifolds} \label{subsec:Qman}
We will assume that the reader has a grasp of the basic theory of supermanifolds. A  \emph{supermanifold} $M := (|M|, \:  \cO_{M})$ of dimension $n|m$, we understand to be a supermanifold as defined by Berezin \& Leites  \cite{Berezin:1975,Leites:1980}. That is, a supermanifold is  a locally superringed space that is locally isomorphic to $\mathbb{R}^{n|m} := \big (\R^{n}, C^{\infty}_{\R^{n}}(-)\otimes \Lambda(\zx^{1}, \cdots \zx^{m}) \big)$. Here, $C^{\infty}_{\R^{n}}(-)$ is the sheaf of smooth real functions on $\R^n$ and $\Lambda(\zx^{1}, \cdots \zx^{m})$ is the Grassmann algebra (over $\R$) with $m$ generators.  The underlying topological space $|M|$ is, in fact, a smooth manifold. This manifold we refer to as the \emph{reduced manifold}. Morphisms of supermanifolds are morphisms as superringed spaces. That is, a morphism $\varphi : M \rightarrow N$ consists of a pair $ \varphi = (|\varphi|, \varphi^*  )$, where $|\varphi| : |M| \rightarrow |N|$ is a continuous map (in fact, smooth) and  $\varphi^*$  is a family of superring homomorphisms $\varphi^*_{|V|} : \cO_N(|V|) \rightarrow \cO_M\big( |\varphi|^{-1}(|V|)\big)$, for every open $|V| \subset |N|$, that respect the restriction maps. The category of supermanifolds we denote as $\catname{SMan}$. We will set $C^\infty(M) := \cO_M(|M|)$. An important and useful result is that 
$$\Hom_{\catname{SMan}}(M_1, 
M_2) \cong \Hom_{\catname{SAlg}}\big(C^\infty(M_2), C^\infty(M_1)\big)\,.$$
where $\catname{SAlg}$ is the category of unital, supercommutative, associative algebras. This allows one to understand morphisms of supermanifolds in terms of homomorphisms between the algebras of global sections. Via minor abuse of notation, we set $\varphi^* : C^\infty(M_2) \rightarrow C^\infty(M_1)$ for the algebra homomorphism associated with a supermanifold homomorphism.  For an introduction to supermanifolds, the reader may consult Carmeli et al. \cite{Carmeli:2011}.
\par 
A \emph{Q-manifold}, following Schwarz \cite{Schwarz:1993} circa late 1992, is a supermanifold equipped with an odd vector field $Q \in \Vect(M)$ that satisfies the integrability condition $ [Q,Q]=0$.  Such vector fields are referred to as \emph{homological vector fields} as $(C^\infty(M), Q)$ is a differential algebra. We remark that non-trivial homological vector fields only exist on supermanifolds with non-zero odd dimension, i.e., that cannot exist on pure even manifolds. It should be noted that the notion of a homological vector field on a supermanifold goes back to at least 1980 and is attributed to  Shander \cite{Shander:1980}. We also highlight the foundational contribution of Vaintrob \cite{Vaintrob:1996}. Note that the action of $Q$ extends from $C^\infty(M)$ to the tensor algebra via the Lie derivative $\mathcal{L}_Q$.\par 
\emph{Morphisms of Q-manifolds}, which we refer to as Q-morphisms, are morphisms of supermanifolds $\varphi : M_1 \rightarrow M_2$, such that 
$$Q_1 \circ \varphi^* -\varphi^* \circ Q_2 =0\,.$$
We say that $Q_1$ and $Q_2$ are $\varphi$-related. In this way we obtain the category of Q-manifolds, which we will denote as $\catname{QMan}$. Isomorphisms in the category of Q-manifolds are obviously defined, and we refer to them as Q-diffeomorphisms.
\begin{example}
Any supermanifold or manifold may be equipped with the zero homological vector field.  We will refer to such Q-manifolds as trivial Q-manifolds.
\end{example}
\begin{example}
Let $N$ be a  pure even manifold; then we can identify (generally inhomogeneous) differential forms as functions on the odd tangent bundle $\Pi \sT N$. That is, we have the natural isomorphism $\Omega^\bullet(N) \cong C^\infty(\Pi \sT N)$. The de Rham differential is identified with a homological vector field and so $(\Pi \sT N, \rmd)$ is a Q-manifold. 
 \end{example}
\begin{example}
By replacing the tangent bundle with a Lie algebroid $(A, [-,-], \rho)$, we obtain the Q-manifold $(\Pi A, \rmd_A)$, where $\rmd_A$ is the Chevalley--Eilenberg differential of the Lie algebroid.  See Vaintrob \cite{Vaintrob:1997} for details.
\end{example} 
\begin{example}
Let $(N, P)$ be a Poisson manifold. We can interpret the Poisson bivector as a function on the odd cotangent bundle $\Pi \sT^* N$, and the Schouten--Nijenhuis bracket is identified with the canonical odd Poisson bracket. The Poisson condition $\SN{P,P}_{\textrm{SN}} =0$, implies that $\delta := \SN{P,-}_{\textrm{SN}}$ is a homological vector field. The homological vector field $\delta$ is identified with the Lichnerowicz differential. For an introduction to Poisson geometry, the reader may consult \cite{Crainic:2021}. 
\end{example} 
We remark that the category $\catname{QMan}$ is a Cartesian monoidal category (see Mac Lane \cite[Section VII]{MacLane:1988}).  If $(M_1, Q_1)$ and $(M_2, Q_2)$ are Q-manifolds then $(M :=M_1 \times M_2,  Q :=Q_1 + Q_2)$ is again a Q-manifold, and the monoidal unit is $(\R^{0|0}, \mathbf{0})$.\par
A homological vector field on a supermanifold we can interpret as a smooth action of the additive supergroup $\R^{0|1}$, which we equip with the odd coordinate $\theta$. The action $\mathsf{a}: \R^{0|1}\times M \rightarrow M$  corresponds to the algebra map $\mathsf{a}^* : C^\infty(M) \rightarrow \Lambda(\theta)\otimes C^\infty(M)$, which we write as
\begin{equation}
\mathsf{a}^* f := f + \theta \, Q(f)\,,
\end{equation}
for any $f \in C^\infty(M)$.  One can check that the group actions axioms hold provided $Q^2 =0$.  We then interpret this action as the ``odd flow'' generated by $Q$. In more physical language, we think of BRST transformations, and $\theta$ plays the role of an ``odd gauge parameter''.  For details of Lie supergroups and actions, see Carmeli et al. \cite[Chapter 7 and Chapter 8]{Carmeli:2011} \par
The set of Q-closed functions (Q-cocycles) we denote by $\mathcal{Z}(M,Q) := \ker(Q)$, and the set of Q-exact functions  (Q-boundaries) as $\mathcal{B}(M,Q):=  \textrm{Im}(Q)$. The \emph{standard cohomology}\footnote{In the BRST formalism, the standard cohomology is identified (modulo ghost number) as the BRST-cohomology.} of $(M,Q)$ is defined as $H^\bullet_\textnormal{std}(M,Q) :=  \mathcal{Z}(M,Q)/ \mathcal{B}(M,Q)$. Note that we have a two-term complex labelled by the Grassmann parity.  For more information about the standard cohomology and characteristic classes, the reader may consult Lyakhovich et al. \cite{Lyakhovich:2010}. For an overview of homological algebra, the reader may consult  Weibel \cite{Weibel:1994}. 
\begin{example}
Let $N$ be a pure even manifold; then the standard cohomology of $(\Pi \sT N, \rmd)$ is identified with the ($\Z_2$-graded) de Rham cohomology of $N$. 
\end{example}
In practice, the structure sheaf of a Q-manifold is often $\Z_2 \times \mathbb{N}$ or $\Z_2 \times \Z$ graded; the additional grading, which we refer to as \emph{weight}\footnote{More generally, the structure sheaf is $\Z_2 \times \mathbb{N}^n$  or $\Z_2 \times \Z^n$ graded and we have a \emph{multi-weight}. For example, in the BV--BRST formalism there are ghost and antifield numbers, etc. to contend with.}. The homological vector field is usually of weight $+1$. In  this paper, weight will play no explicit role. Thus, we will suppress the weight and work with the $\Z_2$-grading only, which, in general, is independent of the weight. For example, within the BRST formalism applied to gauge mechanics, the weight may be identified with ghost number.\par 
We will drop explicit reference to super, and understand all objects to be $\Z_2$-graded.  For example, by associative algebra we will mean associative superalgebra, and by Lie algebra we will mean Lie superalgebra, etc.
\begin{warning}
We will, as common in supergeometry, use $\R^{n|m}$ to denote the supermanifold $\big (\R^{n}, C^{\infty}_{\R^{n}}(-)\otimes \Lambda(\zx^{1}, \cdots \zx^{m}) \big))$ as well as the $\Z_2$-graded vector space  $\R^n\oplus \R^m$. The context should make the meaning clear.
\end{warning} 
%
%
\section{The Derived Algebra of a Q-manifold}\label{sec:DerAlg}
\subsection{The Derived Associative Algebra}\label{subsec:ShiftAssAlg}
Following Loday \cite{Loday:2001} (also see Uchino \cite{Uchino:2009} for higher versions, and Bruce \cite[Section 6]{Bruce:2014} for earlier work on the derived product on a Q-manifold), we make the following definition, taking into account the Grassmann parity of functions on a Q-manifold.
\begin{definition}\label{def:ShiftedMult}
Let $(M, Q)$ be a Q-manifold. Then the \emph{derived product/multiplication} is the $\R$-bilinear map 
$$\star  : C^\infty(M) \times C^\infty(M) \longrightarrow C^\infty(M)\,,$$
defined as $f \star g := (-1)^{\widetilde{f}}\, Q(f)\,g$, for all $f,g \in C^\infty(M)$.
\end{definition}
\begin{remark}
Derived products should be viewed as the associative cousin of derived brackets; see Kosmann-Schwarzbach \cite{Kosmann-Schwarzbach:2004} and references therein.  For example, it is well known that on any Q-manifold $[X,Y]_Q : = (-1)^{\widetilde{X}}\, [ [Q, X],Y]$ defines an odd Loday--Leibniz bracket on $\Vect(M)$.
\end{remark}
\begin{proposition}
Let $(M,Q)$ be a Q-manifold; the derived multiplication (see Definition \ref{def:ShiftedMult}) is 
\begin{enumerate}[i)]
\item Grassmann odd, i.e., $\widetilde{f\star g} = \widetilde{f} + \widetilde{g} +1$; and
\item Associative, i.e.,  $(f \star g )\star h = f \star (g \star h)$.
\end{enumerate}
\end{proposition}
\begin{proof}\
\begin{enumerate}[i)]
\item This is clear as $\widetilde{Q} =1$.
\item This requires a short calculation.
$$
(f \star g) \star h  = (-1)^{\widetilde{f} + \widetilde{f} + 1 + \widetilde{g}} \, Q\big( Q(f)g \big)h = (-1)^{\widetilde{f} + \widetilde{g}}\, Q(f)Q(g)h = f \star(g \star h)\,,
$$
with $f,g,h\in C^\infty(M)$. Note we have used $Q^2 =0$. 
\end{enumerate}
\end{proof}
We will refer to $(C^\infty(M), + , \star)$ as the \emph{derived associative algebra} of $(M,Q)$, noting that under the shifted grading $\overline{f} := \widetilde{f} +1$ we have a standard $\Z_2$-graded associative algebra.  In particular, $\widetilde{f \star g} = \widetilde{f} + \widetilde{g} + 1$, which implies 
\begin{equation}
\overline{f \star g}= \overline{f} +1 + \overline{g} + 1  = \overline{f} + \overline{g}\,.
\end{equation} 
That is, with respect to the shifted grading, the derived product is even. For brevity, we will on occasion write  $(C^\infty(M) , \star)$ or $(C^\infty(M) , \star_Q)$ for the derived associative algebra.\par  
Recall that $C^\infty(M)$ with its usual product is a nuclear Fre\'{chet} algebra (see Carmeli et al. \cite[Chapter 4.5]{Carmeli:2011} for details). Specifically, the usual product is jointly continuous. As $Q$ is a derivation of sections of the structure sheaf, the map $f \mapsto Q(f)$ is continuous with respect to the standard  Fre\'{chet} topology on $C^\infty(M)$. Thus, the derived product $f\star g = \pm \, Q(f) \, g$ is jointly continuous, and so the derived associative algebra is a nuclear Fre\'{chet} algebra. For details of functional analysis, we defer to Trèves \cite{Treves:2006}.
\begin{remark}
In light of Loday's dialgebras, we may consider another derived associative product given by $f\star' g := f \, Q(g)$. However, as the underlying algebra is supercommutative, a quick calculation shows that $f \star' g = -(-1)^{\bar{f}\, \bar{g}}\, g \star f $. Thus, these two derived associative algebras are (shifted) anti-opposite algebras. 
\end{remark}
\begin{warning}
We are \emph{not} employing the parity reversion functor and considering $\Pi C^\infty(M)$ (as the underlying vector space). One may use the parity reversion functor to exchange even and odd forms of associative multiplication. However, attempting to do so here would destroy the geometric nature of the constructions, i.e., the algebraic structure is not guaranteed to be the algebra of functions on a supermanifold. Rather, we are shifting the labelling of the grading to allow direct application of standard ideas from $\Z_2$-graded algebra.
\end{warning}
It is important to note that the derived associative algebra is not (shifted or otherwise) graded commutative.  The noncommutativity is clear due to the `lopsided' definition of the derived product.
\begin{example}
Let $M$ be a (super)manifold considered as a trivial Q-manifold, i.e., the homological vector field is the zero vector field. In this case, the shifted product defines a zero or null associative algebra. That is,  $f \star g =0$ for all functions $f,g \in C^\infty(M)$.
\end{example}
\begin{example}\label{exa:HigherPoisson}
A finite-dimensional version of the classical BV formalism is modelled by higher Poisson structures (see \cite{Voronov:2005}).  A \emph{higher Poisson manifold} is a supermanifold $M$ equipped with an even multivector field $\mathcal{P}\in C^\infty(\Pi \sT^* M)$ such that $\SN{\mathcal{P}, \mathcal{P}} =0$, where $\SN{-,-}$ is the canonical odd Poisson bracket on $\Pi \sT^* M$: in the classical setting the odd Poisson bracket is identified with the Schouten--Nijenhuis bracket. The higher Poisson structure $\mathcal{P}$ is interpreted as a `classical extended action '; the odd symplectic self-commuting condition is interpreted as the `classical master equation', and the bracket itself is interpreted as the antibracket. The derived product is given by $\mathcal{X} \star \mathcal{Y} = (-1)^{\widetilde{\mathcal{X}}}\, \SN{\mathcal{P}, \mathcal{X}}\, \mathcal{Y}$, for all $\mathcal{X}, \mathcal{Y} \in C^\infty(\Pi \sT^* M)$.
\end{example}
\begin{example}
Let $N$ be a smooth manifold; then differential forms on $N$ are understood as functions on  $M := \Pi \sT N$ and the de Rham differential is understood as a homological vector field. Thus, $\Omega^\bullet(N) \cong C^\infty(\Pi \sT N)$. The shifted product is 
$$\alpha \star \beta = (-1)^{\widetilde{\alpha}}\, \rmd \alpha \wedge \beta \,. $$
Suppose $\alpha$ is a nowhere vanishing one-form.  Then via the Frobenius Theorem, we note that if
$$\alpha \star \alpha = - \rmd \alpha \wedge \alpha =0\,,$$ 
then $\ker(\alpha)$ defines a completely integrable distribution. Generalising this, suppose we have $r< \dim N$ (point-wise) linearly independent one-forms  $\{\alpha^1, \alpha^2 , \cdots , \alpha^r \}$. Let us define $\Omega := \alpha^1 \wedge \alpha^2 \wedge \cdots \wedge \alpha^r$. The  Frobenius  Theorem can then be stated as 
$$\alpha^i \star \Omega =0\,,$$
with $1\leq i \leq r$.
\end{example}
\begin{example}\label{exa:GauSys}
In the theory of (finite-dimensional) gauge systems, the BRST differential is recognised as a homological vector field on the supermanifold of configurations, ghosts and antifields, etc.  A little more carefully, we have a triple $\big(M, \{ -,-\}_\epsilon, \mathcal{Q}_{\epsilon+1}\big )$ where the Poisson bracket is of Grassmann parity $\epsilon \in \Z_2$, and the homological potential $\mathcal{Q}_{\epsilon+1} \in C^\infty(M)$ is of Grassmann parity $\epsilon+1$.  The BRST differential is defined as $Q = \{\mathcal{Q}_{\epsilon+1}, - \}_\epsilon$, which requires the classical master equation $\{\mathcal{Q}_{\epsilon+1}, \mathcal{Q}_{\epsilon+1} \}_\epsilon =0$ to hold. For details, the reader may consult \cite{Henneaux:1992} and/or \cite{Lyakhovich:2004}. The derived product is thus
$$f \star g = (-1)^{\widetilde{f}}\, \{\mathcal{Q}_{\epsilon+1}, f \}\, g = (-1)^{\epsilon+1}\, X_f(\mathcal{Q}_{\epsilon +1})\, g\,,$$
where we define the Hamiltonian vector field $X_f := \{f, - \}_\epsilon$. We remark that if $f$ and $\mathcal{Q}_{\epsilon +1}$ are in involution, then $f \star g =0$.
\end{example} 
\begin{proposition}\label{prop:QMorph}
Let $(M_1, Q_1)$ and $(M_2, Q_2)$ be Q-manifolds and $\varphi : M_1 \rightarrow M_2$ be a morphism of Q-manifolds. Then $\varphi^* : C^\infty(M_2) \rightarrow C^\infty(M_1)$ is a derived associative algebra homomorphism.
\end{proposition}
\begin{proof}
Linearity is automatic by the definition of the pullback. We only need to check the behaviour of the pullback on the derived product of two functions. 
\begin{align*}
\varphi^*(f \star_2 g) &= \varphi^* \big( (-1)^{\widetilde{f}}\, Q_2(f)g\big) = (-1)^{\widetilde{f}}\, \varphi^*\big( Q_2(f)\big)\varphi^*g \\
&= (-1)^{\widetilde{f}}\, Q_1(\varphi^*f)\varphi^*g = \varphi^*f \star_1 \varphi^*g\,,
\end{align*}
where we have used $Q_1 \circ \varphi^* = \varphi^* \circ Q_2$.
\end{proof}
Proposition \ref{prop:QMorph} tells us that we have a functor $\cO_{\textrm{der}}: \catname{QMan}^{\textrm{op}} \rightarrow \catname{NUSAlg}$  (non-unital associative superalgebras) via the construction of the derived associative algebra. 
\begin{example}
A morphism of gauge systems (see Example \ref{exa:GauSys}) $\big(M, \{ -,-\}_\epsilon , \mathcal{Q}_{\epsilon+1} \big) \rightarrow  \big(M', \{ -,-\}'_\epsilon , \mathcal{Q}'_{\epsilon+1} \big)$ consists of a smooth map $\varphi : M \rightarrow M'$ that satisfies
$$\varphi^*\big(\{f,g \}'_\epsilon \big)= \{ \varphi^*f , \varphi^* g\}_\epsilon\,,\qquad \varphi^* \mathcal{Q}' = \mathcal{Q}\,,$$
for all $f,g \in C^\infty(M')$. Directly we have 
$$\varphi^*\big(Q'(f) \big) = \varphi^* \big(\{ \mathcal{Q}'_{\epsilon +1}, f\}'_\epsilon \big) = \{\varphi^* \mathcal{Q}'_{\epsilon +1}, \varphi^* f \}_\epsilon = \{\mathcal{Q} \varphi* f\}_\epsilon= Q( \varphi^*f )\,.$$
Thus, morphism of gauge systems induce morphisms between the respective derived associative algebras.  
\end{example}
\begin{example}\label{exa:OddAnchor} Given a Q-manifold $(M, Q)$, there is an `odd anchor' map $\mathsf{a}_Q: M \rightarrow \Pi \sT M $, given in local coordinates as $\mathsf{a}^*_Q\big(x^a, \rmd x^b\big) = \big(x^a, Q^b(x)\big)$, where the homological vector field is $Q= Q^a(x)\partial_a$. The odd tangent bundle is itself a Q-manifold; functions on $\Pi \sT M$ are defined as pseudo-differential forms on $M$ and the homological vector field is the de Rham differential $\rmd = \rmd x^a \partial_a$. Thus, we set $\Omega^\bullet(M):=  C^\infty(\Pi \sT M)$.  It can be shown, see \cite[Proposition 2.8]{Bruce:2017}, that $\mathsf{a}_Q: (M,Q) \rightarrow (\Pi \sT M, \rmd)$ is a Q-morphism, i.e., $Q \circ \mathsf{a}^*_Q  - \mathsf{a}^*_Q \circ \rmd =0$. Thus,
$$\mathsf{a}^*_Q : \big(\Omega^\bullet(M), \star_{\rmd} \big) \rightarrow \big( M , \star_Q\big)\,,$$
is a morphism of  derived associative algebras. 
\end{example}
\begin{remark}
Although $\cO : \catname{SMan}^{\textrm{op}}\rightarrow \catname{SAlg}$ is fully faithful, $\cO_{\textrm{der}}: \catname{QMan}^{\textrm{op}} \rightarrow \catname{NUSAlg}$ fails to be full, i.e., there may be more morphisms as algebras than Q-manifolds. This can easily be argued as the derived product does not have a unit that must be respected if the morphism is to come from a supermanifold morphism. 
\end{remark}
\begin{proposition}\label{prop:LAbs}
Let $(M,Q)$ be a Q-manifold. Then the set of left absorbers of the derived product is identified with the set of Q-closed functions $\mathcal{Z}(M)$.
\end{proposition}
\begin{proof}
We  require $f\star g = (-1)^{\widetilde{f}}\, Q(f)g =0$ for all $g\in C^\infty(M)$. Setting $g=1$, we obtain $Q(f)=0$. Thus, $f \in \mathcal{Z}(M)$, as required. 
\end{proof}
\begin{observation}
Let $(M,Q)$ be a Q-manifold such that $Q \neq \boldsymbol{0}$; then the zero function is a right absorber of the derived product.
\end{observation}
\begin{remark}
Due to nilpotent elements, we cannot conclude that the zero function is a unique absorber. For example, consider $\R^{1|1}$ which we equip with coordinates $(x, \theta)$, and the homological vector field $Q = \theta \partial_x$. Then for any $f = f_0 + \theta \, f_1$ and  $f = f_0 + \theta \, f_1$ we have $f\star g = - \theta \,f'_1(g_0 + \theta \, g_1)$. As $\theta^2 =0$, right absorbers are of the form $0 + \theta \, g_1 $, where $g_1$ is an arbitrary (smooth) function of $x$. 
\end{remark}
Note that there is no unit element for the derived product, as $\star$ shifts the Grassmann parity of the functions. Thus, we have no notion of units for the derived product, i.e., there are no invertible functions with respect to the derived product. The adjoining of a unit is discussed in Subsection \ref{subsec:AdjUniEle}.
\begin{remark}
We may have functions (at least locally) such that $f \star g =1$. For example, on $\R^{p|q}$ with coordinates $(x^a, \theta^\alpha)$ equipped with the homological vector field $Q := \partial_{\theta^1}$, the functions $f = -\theta^1$ and $g = 1$ satisfy $f \star g =1$. However, the unit function is not the unit/identity element of the derived associative product. 
\end{remark}
Given the odd nature of the derived product, we make the following definition. 
\begin{definition}\label{def:ShiftCom}
Let $(M, Q)$ be a Q-manifold. The \emph{derived commutator} of $f,g \in C^\infty(M)$ is defined as 
$$[f,g]_\star := f \star g - (-1)^{(\widetilde{f}+1)(\widetilde{g} +1)}\, g \star f = (-1)^{\widetilde{f}}\, Q(fg)\,.$$
Two functions $f$ and $g$ are said to \emph{derived commute} if $[f,g]_\star=0$.
\end{definition}
\begin{observation}
Let $(M,Q)$ be a Q-manifold, then functions $f$ and $g \in C^\infty(M)$ derived commute, i.e., $[f,g]_\star =0$, if and only if $fg \in \mathcal{Z}(M)$. 
\end{observation}
\begin{proposition}
Let $(M, Q)$ be a Q-manifold such that  $Q \neq \mathbf{0}$.  Then  $\mathcal{Z}(M,Q) := \ker(Q)$ and $\mathcal{B}(M,Q) := \mathrm{Im}(Q)$ are two-sided ideals of the derived associative algebra of $(M,Q)$.
\end{proposition}
\begin{proof}
Assume that $f \in \mathcal{Z}(M,Q)$ and that $g \in C^\infty(M)$.  Then $f \star g =0 \in \mathcal{Z}(M,Q)$ as Q-closed functions are left absorbers (see Proposition \ref{prop:LAbs}). Similarly, $g \star f = \pm Q(g)f = \pm Q(gf) \in \mathcal{Z}(M,Q)$ as $f$ is Q-closed. Thus, $\mathcal{Z}(M,Q)$ is a two-sided ideal of the derived associative algebra.    The same reasoning establishes that $\mathcal{B}(M,Q) \subset\mathcal{Z}(M,Q)$ is also a two-sided ideal. 
\end{proof}
\begin{remark}
We may have non-trivial idempotents, i.e., functions such that $f\star f =f$. Note we require  $\widetilde{f} = 1$ (or equivalently $\overline{f} =0$). For example, on $\R^{0|1}$ with a single coordinate $\theta$ and $Q = \partial_\theta$, $f= - \theta$ is an idempotent.
\end{remark}
As the commutator bracket in an associative algebra forms a Poisson algebra, the following result is naturally expected.
\begin{proposition}\label{prop:ShiftCommProps}
Let $(M,Q)$ be a Q-manifold. Then the derived commutator satisfies 
\begin{enumerate}[i)]
\item $[f, [g, h]_\star]_\star = [[f,g]_\star ,h]_\star + (-1)^{(\widetilde{f}+1)(\widetilde{g}+1)}\, [g, [f,h]_\star]_\star$;
\item $[f, g \star h]_\star = [f,g]_\star\star h + (-1)^{(\widetilde{f}+1)(\widetilde{g}+1)}\,g \star [f,h]_\star$, 
\end{enumerate}
for all $f,g$ and $h \in C^\infty(M)$.
\end{proposition}
\begin{proof}
We observe that the derived associative algebra is a standard associative algebra with respect to the shifted grading defined by $\overline{f} = \widetilde{f} +1$.  It is a standard result that the commutator bracket on an associative algebra is a Poisson bracket. Thus, the above properties hold. 
\end{proof}
\begin{remark}
Proposition \ref{prop:ShiftCommProps} can also be directly proved from the definition of the derived commutator and $Q^2=0$. 
\end{remark}
We will refer to $(C^\infty(M), +, \star, [-,-]_\star)$ as the \emph{derived Poisson algebra} of $(M,Q)$, noting that we have a standard Poisson algebra using the shifted grading.  It is important to note that $[f,g]_\star \in \mathcal{Z}(M)$, and so is a left absorber via Proposition \ref{prop:LAbs}, i.e., $[f,g]_\star \star h =0$ for all $f,g$ and $h\in C^\infty(M)$. \par 
Using the shifted grading, we have a standard Lie algebra; that is, the Jacobi identity is simply
\begin{equation}
[f, [g, h]_\star]_\star = [[f,g]_\star ,h]_\star + (-1)^{\overline{f}\,\overline{g}}\, [g, [f,h]_\star]_\star\,.
\end{equation}
%
%
\subsection{Infinitesimal Symmetries and Derived Derivations}\label{subsec:DerivedDerivations}
The definition of an infinitesimal symmetry of a Q-manifold is the following.
\begin{definition}\label{def:SymQ}
An \emph{infinitesimal symmetry} of a Q-manifold $(M,Q)$ is a vector field $X \in \Vect(M)$ such that $\mathcal{L}_X Q = [X,Q] =0$.
\end{definition}
The $\R$-vector space of infinitesimal symmetries forms a Lie algebra with respect to the Lie bracket of vector fields. Geometrically, the infinitesimal change of $Q$ in the `direction'  $X$ is zero whenever we have an infinitesimal symmetry. We will denote the Lie algebra of infinitesimal symmetries as $\big( \mathrm{Sym}(M,Q), [-,-]\big)$.
\begin{example}
Let $f \in \mathcal{Z}(M,Q)$, i.e., $Q(f)=0$. Then $X := f\, Q$ is an infinitesimal symmetry. To see this, we observe that
$$[X,Q] = [f\,Q, Q] = (-1)^{\widetilde{f}}\, [Q, f\, Q] = (-1)^{\widetilde{f}}\, Q(f)\, Q + f\, [Q,Q]=0\,.$$
\end{example}
\begin{example}
Let $N$ be a pure even manifold. Then $(\Pi  \sT N, \rmd )$ is a Q-manifold. Let $u \in \Vect(N)$ be an arbitrary vector field.  Then the Lie derivative  $\mathcal{L}_u $ is a vector field on $\Pi \sT N$. As $[\mathcal{L}_u ,\rmd] =0$, $\mathcal{L}_u$ is an infinitesimal symmetry. Moreover, $[\mathcal{L}_u, \mathcal{L}_v] = \mathcal{L}_{[u,v]}$.
\end{example}
\begin{remark}
Infinitesimal symmetries preserve the odd (generally singular) distribution $\mathcal{D}_Q:= \Span \{ Q \}$.
\end{remark}
The notion of a derivation of the derived associative algebra is clear.  
\begin{definition}
Let $(M, Q)$ be a Q-manifold. A \emph{derived derivation} is an $\R$-linear map $D: C^\infty(M) \rightarrow C^\infty(M)$ that satisfies the shifted Leibniz rule
$$D(f \star g) = D(f)\star g + (-1)^{\overline{D}\, \overline{f}}\, f \star D(g)\,,$$
for all $f,g \in C^\infty(M)$. The $\R$-vector space of derived derivations we denote as $\Der_\star(M)$.
\end{definition} 
Note that as the derived associative algebra of $(M,Q)$ is inherently noncommutative (assuming $Q\neq \boldsymbol{0}$), derived derivations do not form a (left or right) module over the algebra. We define the Lie bracket of derived derivations as expected: 
\begin{equation}
[D_1, D_2]_\star :=  D_1 \circ D_2 - (-1)^{\overline{D}_1 \, \overline{D}_2}\, D_2 \circ D_1\,.
\end{equation}
Note that we have the standard Lie bracket of derivations with respect to the shifted grading.  Thus, we can immediately infer that derived derivations form a Lie algebra. From Proposition \ref{prop:ShiftCommProps}, we have  \emph{inner derived derivations} defined as 
\begin{equation}
D_f :=  [f, -]_\star\,,
\end{equation}
for any $f \in C^\infty(M)$. Note that we have $\overline{D_f} = \widetilde{f}+1 = \overline{f}$. A derived derivation is said to be an \emph{outer derived derivation} if it is not an inner derived derivation. We will denote the space of inner derived derivations as $\textrm{Inn}_\star(M)$. Formally, the outer derived derivations are defined as as $\textrm{Out}_\star(M):= \Der_\star(M)/ \textrm{Inn}_\star(M) $. 
\begin{example}
On any Q-manifold we have $D_1(f) := [1,f]_\star = Q(f)$, which is an inner derived derivation.
\end{example} \par 
Observe that for vector fields $X\in \Vect(M)$ we have $\overline{X} = \widetilde{X}$. To see this,  recall that $\widetilde{X(f)} = \widetilde{X} + \widetilde{f}$. This implies that $\overline{X(f)} = \widetilde{X}+ \widetilde{f} +1$. On the other hand, $\overline{X(f)} = \overline{X} + \overline{f} = \overline{X} + \widetilde{f} +1$. Thus,  $\overline{X} = \widetilde{X}$.
\begin{proposition}\label{prop:VecFieldsDers}
Let $(M,Q)$ be a Q-manifold. The Lie algebra of infinitesimal symmetries $\big( \mathrm{Sym}(M,Q), [-,-]\big)$ naturally embeds into the Lie algebra of derived derivations $\big ( \Der_\star(M), [-,-]_\star \big )$.
\end{proposition}
\begin{proof}
First $[X, Q] =0$ implies that $X \circ Q = (-1)^{\widetilde{X}}\, Q \circ X$. Then directly we have 
\begin{align*}
X (f \star g) &= X \big((-1)^{\widetilde{f}}\, Q(f)g \big)
              = (-1)^{\widetilde{f}}\, X Q(f)g +(-1)^{\widetilde{f} + \widetilde{X} \,                \widetilde{f} + \widetilde{X}}\, Q(f)X(g)\\
              &= (-1)^{\widetilde{f} + \widetilde{X}}\,  QX(f)g +(-1)^{\widetilde{f} + \widetilde{X} \,                \widetilde{f} + \widetilde{X}}\, Q(f)X(g)
               =X(f) \star g + (-1)^{\widetilde{X}(\widetilde{f}+1)}\, f \star X(g) \\
               &=X(f) \star g + (-1)^{\overline{X}\, \overline{f}}\, f \star X(g)\,.
\end{align*}
Thus, there is an inclusion of $\R$-vector spaces $\mathrm{Sym}(M,Q)\hookrightarrow \Der_\star(M)$. We then observe that, by employing the shifted grading, the Lie bracket $[-,-]_\star$ restricted to the image of the inclusion is precisely the standard Lie bracket of vector fields. Thus, the inclusion is at the level of Lie algebras. 
\end{proof}
\begin{example}
Let $N$ be a smooth manifold; then differential forms on $N$ are understood as functions on  $M := \Pi \sT N$, and the de Rham differential is understood as a homological vector field. Thus, $\Omega^\bullet(N) \cong C^\infty(\Pi \sT N)$. The shifted product is $\alpha \star \beta = (-1)^{\widetilde{\alpha}}\, \rmd \alpha \wedge \beta$. We know (see \cite[Chapter II, Section 8]{Kolář:1993}) that all derivations of $\Omega^\bullet(N)$ are of the from $\mathcal{L}_K = [\rmd , i_K]$, where $K \in \Omega^\bullet(N, \sT N)$ is a vector-valued differential form.  Thus, Fr\"oclicher--Nijenhuis operators of the form $\mathcal{L}_K$ are naturally embedded in the Lie algebra of derived derivations. 
\end{example}
%
%
\subsection{Derived Left Modules}\label{subsec:DerLeftMod}
Let $\mathcal{M}$ be a left $C^\infty(M)$-module (with the standard product of functions).  We will write the action as juxtaposition. We will not insist that the module is locally free, though that is the natural case in supergeometry.  
\begin{definition}
Let $(M,Q)$ be a Q-manifold with $Q \neq \boldsymbol{0}$, and let $\mathcal{M}$ be a left $C^\infty(M)$-module. The \emph{derived external composition} is the map $C^\infty(M)\times \mathcal{M} \rightarrow \mathcal{M}$ given by
$$(f, u) \mapsto f\star u := (-1)^{\widetilde{f}} \, Q(f) u\,.$$
\end{definition}
Using the shifted grading, we set $\overline{u} = \widetilde{u}$, and effectively only shift the grading of the functions.  In doing so, $\overline{f \star u} = \overline{f} + \overline{u}$. 
\begin{proposition}
Let $(M,Q)$ be a Q-manifold with $Q \neq \boldsymbol{0}$, and let $\mathcal{M}$ be a left $C^\infty(M)$-module. The derived external composition is an action of the derived associative algebra on $\mathcal{M}$. 
\end{proposition}
\begin{proof}
A direct calculation establishes the proposition.  Explicitly,
$$ (f \star g)\star u  = (-1)^{\widetilde{g} +1}\, Q\big( Q(f) g\big)u = (-1)^{\widetilde{g} +\widetilde{f}}\, Q(f) Q(g) u = f \star (g \star u)\,. $$ 
\end{proof}
Note that as the derived product has no unit element, it is meaningless to speak of (locally) free modules with respect to the derived external composition. 
\begin{observation}
Let $(M,Q)$ be a Q-manifold with $Q \neq \boldsymbol{0}$, and let $\mathcal{M}$ be a left $C^\infty(M)$-module. Then  $\mathcal{Z}(M,Q) := \ker(Q)$ are the annihilators of the derived external composition. 
\end{observation}
\begin{example}
Let $(M,Q)$ be a Q-manifold and consider $\Vect(M)$. The derived external composition is $f\star X := (-1)^{\widetilde{f}}\, Q(f)X$. Now suppose we have an affine connection $\nabla$; then
$$\nabla_Q \big(f \star X \big) = (-1)^{\widetilde{f}}\, \nabla_Q \big(Q(f)X \big)= (-1)^{\overline{f}}\, f \star \nabla_Q X\,.$$
That is, $\nabla_Q$ is a derived linear operator on vector fields. 
\end{example}
%
%
\subsection{Double Q-manifolds and Compatible Derived Associative Algebras}\label{subsec:DoubleQ}
The definition of a double Q-manifold is a supermanifold with a pair of compatible homological vector fields.   More formally, we have the following.
\begin{definition}
A \emph{double Q-manifold} is a triple $(M, Q_1, Q_2)$, where each $(M, Q_i)$ ($i = 1,2$) are Q-manifolds and $[Q_1, Q_2] =0$.
 \end{definition}
Within physics, we may interpret, modulo (multi-)weight, the pair of homological vector fields within the BRST-anti-BRST formalism as the BRST-differential and the anti-BRST-differential (see, for example, \cite{Bonora:2021}).\par 
On any double Q-manifold we have a pair of derived associative algebras defined in the obvious way, i.e., 
\begin{equation}
f\star_i g := (-1)^{\widetilde{f}}\, Q_i(f)g\,,
\end{equation}
with $i = 1,2$. \par
Recall that a pair of associative algebras on the same vector space is said to be a pair of \emph{compatible associative algebras} if any linear combination of their binary products defines another associative product. 
\begin{proposition}\label{prop:ComParAlgs}
Let $(M, Q_1, Q_2)$ be a double Q-manifold. The pair of associated derived associative algebras is a pair of compatible (derived) associative algebras. 
\end{proposition}
\begin{proof}
We define $f \star g := a\,(f  \star_1 g ) + b\,(f \star_2g)$, with $a$ and $b \in \R$ arbitrary. Clearly, this is just the derived multiplication defined by
$$Q := a \, Q_1 + b\, Q_2\,.$$
The condition $[Q_1, Q_2] =0$ ensures that $Q$ is a homological vector field, and so $\star$ defines a derived associative multiplication. Thus, the pair of derived associative algebras is a pair of compatible associative algebras. 
\end{proof}
One can directly deduce from the definitions the standard compatibility condition 
$$(f \star_1 g)\star_2 h + (f \star_2 g)\star_1 h = f \star_1 (g \star_2 h) + f \star_2 (g \star_1 h)\,.$$
\begin{example}
Any Q-manifold $(M, Q)$ may be considered as a double Q-manifold by taking the second homological vector field to be the zero vector field. 
\end{example}
\begin{example}
Let $(M, Q)$ be a Q-manifold. Pseudo-differential forms on $M$ are understood as functions on $\Pi \sT M$, which itself is canonically a Q-manifold with the homological vector field being identified with the de Rham differential $\rmd$.  The action of $Q$ may be lifted to act on pseudo-differential forms via the Lie derivative $\mathcal{L}_Q := [\rmd , i_Q]$, which is recognised as a homological vector field.  Cartan's formula informs us that $[\rmd , \mathcal{L}_Q]=0$, and so $(\Pi \sT M, \rmd , \mathcal{L}_Q)$ is a double Q-manifold.  Thus, the pair of derived associative algebras associated with the de Rham and the Lie derivative with respect to $Q$ are compatible associative algebras. 
\end{example}
\begin{example}
Following Voronov \cite{Voronov:2012}, a double vector bundle $D$ is a double Lie algebroid, in the sense of Mackenzie \cite{Mackenzie:1992}, if the total parity reversed double vector bundle $\Pi^2 D$ is a double Q-manifold, where the homological vector fields $Q_{(0,1)}$ and $Q_{(1,0)}$ are of bi-weight  $(0,1)$ and $(1,0)$, respectively. Neglecting the bi-weight, the pair of derived associative algebras associated with the homological vector fields are compatible associative algebras.
\end{example}
\begin{remark}
The results of this subsection directly generalise to $n$-tuple Q-manifolds, i,e., supermanifolds equipped with $n$ homological vector fields that pairwise commute. For example, $n$-fold Lie algebroids can be formulated this way; see Voronov \cite{Voronov:2012}.
\end{remark}
%
%
\subsection{Derived Associative Algebras of Q-groups}
The definition we employ for a Q-group (see Mehta \cite{Mehta:2009}) is the following.
\begin{definition}
A \emph{Q-group} $(G, Q)$ is a group object in the category of Q-manifolds.
\end{definition}
For group objects in a category, see Mac Lane \cite[Chapter III.6]{MacLane:1988} and for Lie supergroups in particular, see \cite[Chapter 7]{Carmeli:2011}.  To spell a Q-group out,  we first equip $G \times G$ with the homological vector field  $Q^{(2)} := Q \otimes \mathsf{id} + \mathsf{id} \otimes Q$, and $\R^{0|0}$ with the zero homological vector field $\mathbf{0}$. The structure maps of the Lie supergroup
$$\mu : G \times G \rightarrow G\,, \qquad i : G \rightarrow G\,, \qquad e : \R^{0|0} \rightarrow G\,,$$
referred to as the \emph{multiplication}, \emph{inverse}, and \emph{unit} maps, respectively, must be  morphisms of Q-manifolds. We thus have the compatibility conditions
\begin{subequations}
\begin{align}\label{eqn:ComConQGrp1}
& Q^{(2)}\circ \mu^* = \mu^* \circ Q\,,\\\label{eqn:ComConQGrp2}
& Q \circ i^* = i^* \circ Q \, ,\\ \label{eqn:ComConQGrp3} 
& e^* \circ Q =\mathbf{0}\,.
\end{align}
\end{subequations}
Just as in the classical case for multiplicative vector fields on a Lie group, it can be shown that \eqref{eqn:ComConQGrp1} implies both \eqref{eqn:ComConQGrp2} and  \eqref{eqn:ComConQGrp3} (see Jubin et al. \cite[Proposition 2.20]{Jubin:2020}). Thus, we will refer to $Q$ as a \emph{multiplicative homological vector field}.  
\begin{example}
Any Lie supergroup $G$ may be equipped with the zero homological vector field $\mathbf{0}$; then evidently $(G,\mathbf{0})$ is a Q-group.
\end{example}
\begin{proposition}\label{prop:QGrpAlg}
Let $(G,Q)$ be a Q-group. Then the following are morphisms of derived associative algebras: 
\begin{enumerate}[i)]
\item $\mu^* : \big( C^\infty(G), \star_Q\big)\rightarrow \big( C^\infty(G \times G), \star_{Q^{(2)}}\big)$\,;
\item $i^* :  \big( C^\infty(G), \star_Q\big) \rightarrow  \big( C^\infty(G), \star_Q\big)$\,;
\item $ e^*: \big( C^\infty(G), \star_Q\big) \rightarrow \big( \R, \star_{\boldsymbol{0}}\big)$\,.
\end{enumerate}
\end{proposition}
\begin{proof}
 This follows directly from the fact that $Q$ is a multiplicative homological vector field and Proposition \ref{prop:QMorph}, which states that morphisms of Q-manifolds induce morphisms between their respective derived associative algebras.
 \end{proof}  
\begin{example}
Consider supertranslation group $\R^{1|1}$ which we equip with coordinates $(x, \theta)$. The multiplication is given by
$$\mu^* x  = x_1 +x_2 + \theta_1 \theta_2\,, \qquad \mu^* \theta = \theta_1 + \theta_2\,.$$
The reader can quickly deduce the identity and inverse maps.  The multiplicative homological vector field we claim is $Q = \theta \partial_x$. Note that the only non-zero derived product of the coordinates is  $x \star x  = \theta x$. \par 
Checking multiplicity explicitly, 
\begin{align*}
&  \mu^*\big(Q(x) \big) = \mu^*\theta = \theta_1 + \theta_2\, && Q^{(2)}(\mu^*x ) = (\theta_1 \partial_{x_1} + \theta_2 \partial_{x_2})(x_1 +x_2 + \theta_1 \theta_2) = \theta_1 + \theta_2\,, \\
& \mu^* \big( Q(\theta)\big) = \mu^* 0 =0\,, && Q^{(2)}( \mu^* \theta) = (\theta_1 \partial_{x_1} + \theta_2 \partial_{x_2})(\theta_1 + \theta_2 ) =0\,.
\end{align*}
Thus $Q$ is indeed a multiplicative homological vector field.  It is easy to see that $i^*x =-x$, $i^* \theta =-\theta$, and $e^* x=0$ and $e^* \theta=0$.  \par 
To check that $\mu^*$ is a morphism of derived associative algebras, we first calculate the following.
$$\mu^*(x \star x) = \mu^*(\theta x) = \mu^*\theta \, \mu^* x  = (\theta_1 + \theta_2)(x_1 +x_2 + \theta_1 \theta_2) = \theta_1 x_1 + \theta_1 x_2 + \theta_2 x_1 + \theta_2 x_2\,.$$
Next observe that $Q^{(2)}(\mu^*x) = \theta_1 + \theta_2$, where hence,
$$\mu^* x \star \mu^*x =  (\theta_1 + \theta_2)(x_1 + x_2 + \theta_1 \theta_2) = \theta_1 x_1 + \theta_1 x_2 + \theta_2 x_1 + \theta_2 x_2\,.$$
Thus, as expected, $\mu^*(x \star x) = \mu^* x \star \mu^*x$. 
\end{example}
\begin{warning}
It is \ul{not} the case that $\big( C^\infty(G \times G),\star_{Q^(2)} \big)\cong \big( C^\infty(M), \star_Q \big)\, \widehat{\otimes}\, \big( C^\infty(G), \star_Q\big)$ as  associative (shifted) superalgebras. While the category $\catname{Qman}$ is a Cartesian monoidal category, this does not imply that the derived product is compatible with the (completed) tensor product. We will further comment on this in Section \ref{sec:ConRem}. 
\end{warning}

%
%
\subsection{Adjoining a Unit Element}\label{subsec:AdjUniEle}
The derived associative algebra of a Q-manifold is non-unital. By employing the shifted grading, the standard method of adjoining a unit, also known as the Dorroh extension, can be applied (see \cite{Dorroh:1932}).  The \emph{(minimal) unitalisation of the derived associative algebra} is defined as follows.  First, the (nuclear Fr\'{e}chet) super vector space is 
\begin{equation}
C^\infty(M)^+ := C^\infty(M)\oplus \R \ni (f, a) \,.
\end{equation}
Using the shifted grading, we define $\overline{a} =0$ for all $a \in \R$.  This then implies we have 
\begin{align*}
& (f, a) \in C^\infty_0(M)^+ ~~ ,\textnormal{if} ~ \overline{f} = 0\,,\\
&( f, 0) \in C^\infty_1(M)^+ ~~ ,\textnormal{if} ~ \overline{f} = 1\,,
\end{align*}
The derived product is modified as
\begin{equation}
(f, a) \star (g, b) :=  (f \star g + a\, g + b\, f\,, ab)\,.
\end{equation}
Evidently, the unit element is given by $e := (0, 1)$.  The resulting algebra is a unital associative noncommutative superalgebra.  We stress that the unitalisation of the derived associative algebra is not a derived associative algebra of a Q-manifold. For instance, we cannot have a unit element for a derived product. Morphisms of Q-manifolds extend to the unitalisation as $\varphi^+(f, a) := (\varphi^*f, a)$.\par
In short, we have the composition of functors  that take a Q-manifold and produce a unital nuclear Fr\'{e}chet superalgebra 
$$(M,Q) \rightarrow (C^\infty(M), \star) \rightarrow (C^\infty(M)^+, \star)\,.$$
We will further comment on the unitalisation in Section \ref{sec:ConRem} during the concluding remarks. 
%
%
\section{Derived Cyclic Cohomology}\label{sec:DerCycCoho}
\subsection{The Modular Class and Shifted Graded Traces}\label{subsec:ModClass}
Recall that the modular class of a Q-manifold is the obstruction to the existence of a Q-invariant Berezin volume, see \cite{Bruce:2017,Bruce:2020} for a review.  A Q-manifold is said to be \emph{unimodular} if the modular class vanishes, which by definition means that we can find Berezin volumes $\p$ such that $\mathcal{L}_Q \p =0$. Note that although the modular class requires a choice of Berezin volume to define, it is independent of this choice and has a cohomological understanding. For applications of modular classes of Q-manifolds in the theory of Lie algebroids, the reader may consult Grabowski \cite{Grabowski:2012,Grabowski:2014}.  We remark that within the BV-formalism,  the vanishing of the modular class implies that there is no `one-loop anomaly' in the corresponding semiclassical quantum theory, see \cite[Remark 4.7]{Bruce:2017} and/or \cite[Section 6.1]{Lyakhovich:2004}, for example.
\begin{proposition}\label{prop:BerIntTrace}
Let $(M, Q)$ be a unimodular Q-manifold that is superoriented\footnote{ By a superoriented supermanifold we mean a supermanifold whose reduced manifold is orientable and an orientation has been chosen, and which possesses an atlas  where the standard Jacobian of the even coordinates and the Berezinian of the coordinate transformations are strictly positive. For details of orientations of supermanifolds see Shander \cite{Shander:1988}.} and compact. Then the Berezin integral with respect to a Q-invariant Berezin volume acts as a shifted graded trace on the derived associative algebra.
\end{proposition} 
\begin{proof}
Starting with $\int_M \mathcal{L}_Q(\rho\, f) = \int_M (\mathcal{L}_Q \p)\, f + \int_M \p \, Q(f)$, we observe that if $(M, Q)$ is unimodular and we pick the Q-invariant Berezin volume, this reduces to 
$$\int_M \p \, Q(f) =0\,,$$
where we have used the fact that $M$ is compact.   We then consider $f \mapsto (-1)^{\widetilde{f}}\, fg$, which produces 
$$\int_M \p \, (-1)^{\widetilde{f}}\, Q(fg) = \int_M \p \, [f,g]_\star=0\,.$$
This is equivalent to $\int_M \p \, f \star g = \int_M \p \, (-1)^{\overline{f} \, \overline{g}}\, g \star f$, and hence the integral acts as a shifted graded trace.
\end{proof}
\begin{example}
Let $N$ be a compact and oriented manifold; then $(\Pi \sT N, \rmd)$ is a Q-manifold.  The canonical Berezin volume (locally this is just the coordinate Berezin volume) is invariant under the de Rham differential. Thus, $\Pi \sT N$ is unimodular.  Then $\tau(\alpha) := \int_{\Pi \sT N} \p \, \alpha = \int_N \alpha^{\textrm{top}}$ is a shifted trace over the derived associative algebra given by $\alpha \star \beta = (-1)^{\widetilde{\alpha}}\, \rmd \alpha \wedge \beta$.
\end{example}
As we have established that the Berezin integral can act as a shifted trace, the question of a derived version of Connes' cyclic cohomology arises.
%
%
\subsection{Derived Cyclic Cohomology of Q-manifolds}\label{subsec:ShiftSycCoh} 
In this subsection, we apply cyclic cohomology to Q-manifolds via the derived associative algebra. Recall that cyclic cohomology was developed by Connes \cite{Connes:1985} and acts as a generalisation of de Rham cohomology to associative algebras. For a very accessible introduction, the reader may consult Connes \cite[Chapter 3]{Connes:1994}, Khalkhali \cite{Khalkhali:2010}, or  Weibel \cite[Chapter 9]{Weibel:1994}. Kastler \cite[Chapter 3]{Kastler:1988} developed Hochschild and cyclic cohomology for $\Z_2$-graded associative algebras.  Thus, by employing the shifted grading to the derived product, we may freely employ Kastler's constructions and results with minimal amendment.  \par 
Let  $\catname{NFSVec}$ be the category of nuclear Fr\'{e}chet super vector spaces (see Carmeli et al. \cite[C.2.]{Carmeli:2011}). By definition, morphisms in this category are continuous linear maps. We define $n$-\emph{cochains} as elements of 
$$C^n(M,Q):= \InHom_{\catname{NFSVec}}\big( C^\infty(M)^{\widehat{\otimes}^{n+1}}, \R\big)\,.$$
Note that $C^\infty(M) := \cO_M(|M|)$ is a nuclear Fr\'{e}chet algebra (see \cite[Chapter 4.5]{Carmeli:2011}). Thus, the projective topology is used to construct the completed tensor product in the definition of $n$-cochains. Via Grothendieck's fundamental theorem of nuclear spaces, other `reasonable' topologies can also be used, such as the injective topology; however, all such completed tensor products coincide (see \cite[Chapter 50]{Treves:2006}). 
\begin{warning}
In differential geometry, the space of distributions, understood as dual objects to functions, is defined using functions with compact support. On a compact (super)manifold there is, of course, no difference between $C^\infty(M)$ and $C^\infty_c(M)$. Note that, for non-compact Q-manifolds, we may consider the derived associative algebra of compactly supported functions $(C^\infty_c(M), \star)$. The complication in doing so is that the underlying vector space is an LF-space, i.e., a strict inductive limit of Fr\'{e}chet spaces, rather than a Fr\'{e}chet space itself (see \cite[Chapter 13]{Treves:2006}). While it remains a nuclear space, working with completed projective tensor products of LF-spaces is significantly more complicated than the Fr\'{e}chet case. Thus, to avoid complications with defining $n$-cochains ($n>0$), we will not insist on compactly supported functions in the definition, and in practice will work with compact supermanifolds.     
\end{warning}
It is irrelevant if one considers the standard or shifted grading with the $n$-cochains, as long as one is consistent. Conventionally, we will consider the standard Grassmann parity.   The shifted grading can then be applied when required in the formulas to keep track of the sign factors.  For pure elements we will write $\phi^n(f_0, f_1, \cdots, f_n) := \phi^n(f_0 \otimes f_1 \otimes \cdots \otimes f_n)\in C^n(M,Q)$. In general, a cochain will be inhomogeneous in grading. We view $n$-cochains as generalised currents via direct analogy with de Rham currents\footnote{See Current, \href{ http://encyclopediaofmath.org/index.php?title=Current&oldid=56194}{\emph{Encyclopedia of Mathematics}}, http://encyclopediaofmath.org/index.php?title=Current\&oldid=56194, (accessed 16/07/2026)}. \par 
The \emph{shifted cyclic permutation operator} $\lambda   \, : C^n(M,Q) \longrightarrow C^n(M,Q)$ is defined as  
$$(\lambda \phi^n)(f_0, f_1, \cdots , f_n) := (-1)^{ n+ \overline{f}_n\, (\overline{f}_0 + \cdots +\overline{f}_{n-1})}\, \phi^n(f_n, f_0 , \cdots , f_{n-1})\,. $$
The space of \emph{derived cyclic n-cochains} $C^n_\lambda(M,Q) \subset C^n(M,Q)$ consists of $n$-cochains that are invariant under $\lambda$, i.e., $\lambda \phi^n = \phi^n$. \par 
The \emph{Hochschild coboundary operator} (see \cite{Kastler:1988}) $ b :  C^n(M, Q) \longrightarrow C^{n+1}(M,Q)$ is defined as
\begin{align}\label{eqn:HochCobOp}
  (b \phi^n)(f_0, f_1 , \cdots , f_{n+1}) & := \sum_{i=0}^n (-1)^i \, \phi^n(f_0, \cdots f_i \star f_{i+1}, \cdots , f_{n+1})\\ \nonumber 
 &+ (-1)^{ n+1 + \overline{f}_{n+1}\, (\overline{f}_0 + \cdots + \overline{f}_n)}\, \phi^n(f_{n+1}\star f_0, \cdots , f_n)\,. 
\end{align} 
The result $b^2 =0$ may be directly established. Moreover, it can be shown that the space of derived cyclic cochains is invariant under the Hochschild coboundary operator, i.e., $b \big( C^n_\lambda(M,Q)\big)\subset C^{n+1}_\lambda(M,Q)$ (see \cite{Kastler:1988}). Note that the Hochschild coboundary operator is itself continuous as it is defined in terms of a continuous linear map and the operations within the derived associative algebra are continuous. 
\begin{observation}
If all $f_i \in C^\infty(M)$ ($i =0, \cdots, n+1$) are Q-exact, i.e., are all in $\mathcal{B}(M,Q)$, then $b \phi^n(f_0, \cdots , f_{n+1}) =0$ for any and all $\phi^n \in C^n(M,Q)$.
 \end{observation}
The \emph{derived cyclic complex}  of $(M,Q)$ is the subcomplex of the derived Hochshild complex  given by 
\begin{equation}\label{eqn:ShiCycCom}
C^0_\lambda(M,Q) \stackrel{b}{\longrightarrow} C^1_\lambda(M,Q) \stackrel{b}{\longrightarrow} C^2_\lambda(M,Q) \stackrel{b}{\longrightarrow} \cdots \, .
\end{equation}
\begin{definition}
The cohomology of the derived cyclic complex \eqref{eqn:ShiCycCom} is referred to as the \emph{derived cyclic cohomology} of the Q-manifold $(M,Q)$, and will be denoted $HC^n_\lambda(M,Q)$, $n=0,1,2, \cdots$. A \emph{derived cyclic n-cocycle} is a $b$-closed derived cyclic $n$-cochain.  The resulting classes in the derived cyclic cohomology we refer to as \emph{derived cyclic classes}, and we denote them as $[\phi^n]$.
\end{definition}
\begin{example}
Consider a trivial Q-manifold $(M,\boldsymbol{0})$. The derived product is a zero product, i.e., $f \star g =0$ for all $f,g \in C^\infty(M)$. In turn,  the  Hochschild coboundary operator $b$ is a zero map, and so all derived cochains are $b$-closed.  Thus, $HC^n_\lambda(M, \boldsymbol{0}) \cong C^n_\lambda(M, \boldsymbol{0})$.  
\end{example}
Suppose $\varphi : (M_1, Q_1) \rightarrow (M_2, Q_2)$ is a morphism of Q-manifolds. Recall that Proposition \ref{prop:QMorph} established that $\varphi^* : C^\infty(M_2) \rightarrow C^\infty(M_1)$ is a morphism of the derived associative algebras. The pushforward of a derived $n$-cochain  $\varphi_* : C^n(M_1, Q_1) \rightarrow C^n(M_2, Q_2)$ is defined as 
\begin{equation}
\big( \varphi_* \phi^n\big)(f_0, \cdots , f_n) := \phi^n(\varphi^*f_0, \cdots , \varphi^* f_n)\,,
\end{equation}
where $f_i \in C^\infty(M_2)$. As $\varphi^*$ does not change the parity (shifted or otherwise), we have $\overline{\varphi^* f} = \overline{f}$; thus the pushforward restricts to derived cyclic $n$-cochains. That is
$$\varphi_* \big( C^n_\lambda(M_1, Q_1)\big)\subseteq  C^n_\lambda(M_2, Q_2) \,.$$
As $\varphi^*$ is a morphism of the derived associative algebras, applying the Hochschild coboundary operator $b$ to $\varphi^*\phi^n$, one can quickly observe that as standard $ b \circ \varphi_* = \varphi_*\circ b$. Thus, the pushforward $\varphi_*$ is a cochain map between the derived cyclic complexes. This implies that the pushforward descends to the cohomology:
\begin{align}
& HC^n_\lambda(M_1, Q_1) \longrightarrow HC^n_\lambda(M_2, Q_2)\\ \nonumber
& [\phi^n] \longmapsto [\varphi_* \phi^n]\,.
\end{align}
Thus, we have established the following. 
\begin{theorem}\label{thm:DerCycCoFunct}
The  assignments $\catname{QMan}\in (M,Q)\mapsto HC_\lambda^n(M,Q)\in \catname{SVec},$ and $\Hom_{\catname{QMan}}\big((M_1, Q_1), (M_2, Q_2)\big)\ni \varphi \mapsto \varphi_* : HC_\lambda^n(M_1,Q_1)\rightarrow HC_\lambda^n(M_2,Q_2)$ define a covariant functor $HC^n_\lambda :\catname{QMan} \rightarrow \catname{SVec}$ ($n=0,1, 2, \cdots$).
\end{theorem}
\begin{remark}
Note that we usually have the cohomology of a space being a contravariant functor. However, the derived cyclic cohomology is a cohomology theory of the algebra of functions on a Q-manifold, and thus a covariant functor when considered as a cohomology theory of Q-manifolds.
\end{remark}
\begin{corollary}\label{cor:Qdiff}
If the Q-manifolds $(M_1, Q_1)$ and $(M_2, Q_2)$ are Q-diffeomorphic, i.e., there is an isomorphism in $\catname{QMan}$ between them, then  $HC^n_\lambda(M_1, Q_1) \cong HC^n_\lambda(M_2, Q_2)$, $n=0,1,2, \cdots$.  Conversely, if any of the derived cyclic cohomology groups are not isomorphic, then the Q-manifolds cannot be Q-diffeomorphic.
\end{corollary}
Theorem \ref{thm:DerCycCoFunct} (together with Corollary \ref{cor:Qdiff}) establishes that derived cyclic cohomology is a genuine algebraic invariant of a Q-manifold. 
\begin{proposition}\label{prop:CanMorphCohom}
Let $(M,Q)$ be a Q-manifold and let $(\Pi \sT M, \rmd)$ be the odd tangent bundle equipped with the de Rham differential. Then there is a canonical morphism 
$$ \Psi : HC^n_\lambda(M, Q) \longrightarrow HC^n_\lambda(\Pi \sT M , \rmd)\,$$
for all $n \in \mathbb{N}$.
\end{proposition}
\begin{proof} 
The  `odd anchor' $\mathsf{a}_Q : \Omega^\bullet(M) \rightarrow C^\infty(M)$, which is canonical and exists for all Q-manifolds, is a morphism between the respective derived associative algebras (see Example \ref{exa:OddAnchor}).  Then, via Theorem \ref{thm:DerCycCoFunct}, we know such a morphism induces a morphism between the derived cyclic cochains that descends to the derived cyclic cohomology.  We then set 
\begin{align*}
 \Psi : & HC^n_\lambda(M, Q) \longrightarrow HC^n_\lambda(\Pi \sT M, \rmd)\\
& [\phi^n] \longmapsto [(\mathsf{a}_Q)_* \phi^n]\,.
\end{align*}
\end{proof}
\begin{definition}
Let $(M, Q)$ be a Q-manifold and let $ \Psi : HC^n_\lambda(M, Q) \longrightarrow HC^n_\lambda(\Pi \sT M , \rmd)$ be the canonical morphism defined in Proposition \ref{prop:CanMorphCohom}. We refer to  $[\phi] \in \textsf{Im}(\Psi) \subset HC^n_\lambda(\Pi \sT M , \rmd)$ as  an \emph{anchored derived cyclic class} or, for brevity, an \emph{anchored class} when no confusion can arise.  
\end{definition}
\begin{proposition}\label{prop:CochainsInv}
Let $(M, Q)$ be a Q-manifold, then all $b$-closed cochains $\phi^n \in C^n(M,Q)$  satisfy $\phi^n \circ (Q \otimes \cdots \otimes Q) =0$.
\end{proposition}
\begin{proof}
Let $f_i \in C^\infty(M)$ for $0\leq i \leq n$ be arbitrary (homogeneous) functions. Directly from the  Hochschild coboundary operator \eqref{eqn:HochCobOp}, setting $f_{n+1} =1$ (the constant unit function) and using $1 \star f_0 =0$ we obtain for $b$-closed cochains
\begin{align*}
 \sum_{i=0}^{n}(-1)^i\, \phi^n (f_0, \cdots , f_i \star f_{i+1}, \cdots ,1)  &= \sum_{i=0}^{n-1}(-1)^i\,\phi^n (f_0, \cdots , f_i \star f_{i+1}, \cdots ,1)\\ & + (-1)^n \,  \phi^n(f_0, f_1 , \cdots ,f_n \star 1)\\
&= \sum_{i=0}^{n-1}(-1)^i\,\phi^n (f_0, \cdots , f_i \star f_{i+1}, \cdots ,1)\\ & + (-1)^{n +\widetilde{f}_n} \,  \phi^n(f_0, f_1 , \cdots ,Q(f_n))  =0\,.
\end{align*}
Then sending $f_i \mapsto Q(f_i)$ for $0 \leq i <n$, noting that $Q(f_i) \star Q(f_{i+1}) =0$, we obtain 
\begin{equation}
\phi^n(Q(f_0),Q(f_1), \cdots , Q(f_n)) =0\,.
\end{equation} 
As the functions are arbitrary, the above can be written as $\phi^n \circ (Q \otimes \cdots  \otimes Q) =0$. 
\end{proof}
Any representative of a derived cyclic class $[\phi^n] \in HC^n_\lambda(M,Q)$ is of the form $\phi^n + b \psi^{n-1}$, where $\psi^{n-1} \in C^{n-1}_\lambda(M,Q)$. The $b$-exact term is, of course, cohomologically trivial.  Thus, elements of  $HC^n_\lambda(M,Q)$ for $n\in\mathbb{N}$ are interpreted as equivalence classes of classically BRST-invariant shifted cyclic multi-functionals.  In particular, once we have found a derived cyclic $n$-cocycle, its derived cyclic class $[\phi^n]$ measures  if $\phi^n$ is cohomologically trivial, i.e., if $\phi^n = b \psi^{n-1}$. As $b \phi^n=0$, Proposition \ref{prop:CochainsInv} tells us that  $\phi^n = b \psi^{n-1}$ is BRST-invariant, however $\psi^{n-1}$ itself need not be BRST-invariant. \par  
Note that the cyclic property has not been used to deduce the invariance, and so the derived Hochschild cohomology groups $HC^n(M,Q)$ have a similar interpretation, however, now without the shifted cyclic property.  \par 
A $0$-cochain, $\tau : C^\infty(M) \rightarrow \R$, is a shifted graded trace if and only if $\tau([f,g]_\star)=0$ for all $f$ and $g\in C^\infty(M)$. Moreover, directly we observe that  $(b \tau)(f_0, f_1) = \tau(f_0 \star f_1) - (-1)^{ \overline{f}_1 \, \overline{f}_0 }\, \tau(f_1 \star f_0)$, and so shifted graded traces define derived cyclic $0$-cocycles (there are no $-1$-cochains and so non-$b$-exactness is guaranteed).
\begin{lemma}\label{lem:0Cocycles}
Let $(M, Q)$ be a Q-manifold. A $0$-cochain $\tau$ is a derived cyclic $0$-cocycle if and only if $\tau \circ Q =0$.
\end{lemma}
\begin{proof}
If $\tau \circ Q =0$, then $\tau(Q(h)) = 0$ for all $h \in C^\infty(M)$. Setting $h = (-1)^{\widetilde{f}}\, fg$ we obtain $\tau((-1)^{\widetilde{f}}\, Q(fg)) = \tau([f,g]_\star)=0$. Thus, $\tau$ is a shifted graded trace and so is a derived cyclic $0$-cocycle. Conversely, if $\tau([f,g]_\star) = 0$ for all $f$ and $g \in C^\infty(M)$, setting $g = 1$ (the constant unit function) we obtain $\tau([f, 1]_\star) = \tau ((-1)^{\widetilde{f}}\, Q(f))=0$. Thus, as $f$ arbitrary, we have $\tau \circ Q =0$. 
\end{proof}
\begin{proposition}\label{prop:IntShiCo}
Let $(M,Q)$ be a unimodular, superoriented, compact Q-manifold, and let $\p$ be a Q-invariant Berezin volume. Then the Berezin  integral $\tau(f) := \int_M \p \, f$ is a derived cyclic $0$-cocycle. 
\end{proposition}
\begin{proof}
Given the conditions on $(M,Q)$, we have $\int_M  \p\, Q(f) =0$, and thus, via Lemma \ref{lem:0Cocycles}, we conclude that the Berezin integral (with respect to a Q-invariant Berezin volume) is a derived cyclic $0$-cocycle.
\end{proof}
\begin{remark}
 Proposition \ref{prop:IntShiCo} can also be seen as a direct corollary of Proposition \ref{prop:BerIntTrace}.
\end{remark}
To set some notation, we define $\mathrm{res}_m : C^\infty(M) \rightarrow \cO_{M,m}$ as the algebra map sending a function to its germ at $m \in |M|$. Similarly, we define $\epsilon_m : \cO_{M,m} \rightarrow  C^\infty_m$, which is heuristically `throwing away the odd coordinates', and $\mathrm{ev}_m : C^\infty_m \rightarrow \R$ is the standard evaluation map at $m$. \par 
The homological vector field $Q \in \Vect(M)$ induces a derivation at the levels of stalks, which we denote as $Q_m : \cO_{M,m} \rightarrow \cO_{M,m}$. We then define the associated (odd) tangent vector at $m \in |M|$ as
$$Q|_m := \mathrm{ev}_m \circ \epsilon_m \circ Q_m \in \sT_m M\,.$$
\begin{proposition}
Let $(M,Q)$ be a Q-manifold and $m \in |M|$ be a point such that $Q|_m =0$. then the $0$-cochain $\tau_m := \mathrm{ev}_m \circ \epsilon_m \circ \mathrm{res}_m$ is a derived cyclic $0$-cocycle.
\end{proposition}
\begin{proof}
Let $f\in C^\infty(M)$ be arbitrary. Then, using the standard properties of sheaves and restrictions, etc., we see that 
\begin{align*}
(\tau_m \circ Q)(f) &= \tau_m(Q(f)) = (\mathrm{ev}_m \circ \epsilon_m \circ \mathrm{res}_m) (Q(f))\\
&= (\mathrm{ev}_m \circ \epsilon_m)(Q_m(f_m)) = (\mathrm{ev}_m \circ \epsilon_m \circ Q_m)(f_m)\\
& = Q|_m f_m =0\,,
\end{align*} 
where the last equality is by hypothesis.  As $f$ is arbitrary, we have $\tau_m \circ Q =0$, and via Lemma \eqref{lem:0Cocycles}, $\tau_m$ is a derived cyclic $0$-cocycle.  
\end{proof}
\begin{example}
Consider a supermanifold $M$ considered as a trivial Q-manifold, i.e., the homological vector field is the zero vector field. Thus, $H_{\mathrm{std}}^\bullet(M, 0)= C^\infty(M)$. Any $0$-cochain satisfies $\tau \circ Q=0$ trivially. Thus, $HC^0_\lambda(M,0) \cong \InHom_{\catname{NFSVec}}\big( C^\infty(M), \R\big) = C^0_\lambda(M,Q)$.
\end{example}
\begin{example}\label{exa:OddLine}
Consider the odd line $M = \R^{0|1}$, which we equip with a global coordinate $\zx$, and a canonical homological vector field given by $Q = \partial_\zx$. Functions on $\R^{0|1}$ are of the form $f = a + b\, \zx$, where $a,b \in \R$. Observe that $\mathcal{Z}(\R^{0|1})\cong \R $ and $\mathcal{B}(\R^{0|1})\cong \R $, and so $H^\bullet_\textnormal{std}(\R^{0|1},\partial_\zx) = 0$.  That is, the standard cohomology is trivial.\par
As $(\R^{0|1}, \partial_\zx)$ is unimodular, superoriented (trivially) and  compact $(|\R^{0|1}| = \mathrm{pt})$,  Proposition \ref{prop:IntShiCo} immediately tells us that 
the standard Berezin integral on $\R^{0|1}$ is a derived cyclic $0$-cocycle.  Observe that
$$\tau(f) = \int_{\R^{0|1}}D[\zx] \, (a + b \, \zx) = b \in \R\,.$$
Thus, $[\tau] \in HC^0_\lambda(\R^{0|1},\partial_\zx)$ is non-trivial.  \par 
Any derived cyclic $0$-cocycle $\tau'$ is a linear functional and so is determined by its outputs on $\{ 1, \zx\}$. We set $\tau'(1) = c_0$ and $\tau'(\zx) = c_1$. Thus, $\tau'(a + b \, \zx) = ac_0 + bc_1$. We know from Lemma \ref{lem:0Cocycles} that $\tau' \circ Q =0$.  In the current setting, we have 
$$\tau' \big(\partial_\zx (a + b \, \zx) \big) = \tau'(b) = b c_0 =0\,.$$
This forces $c_0 =0$: meaning that $\tau'(1)=0$, but $\tau'(\zx) = c_1 \in \R$ is not constrained. Thus, $\tau'(a + b\, \zx) = bc_1 = c_1 \tau(a + b\, \zx)$. That is, any other derived cyclic $0$-cocycle is a scalar multiple of the standard Berezin integral on $\R^{0|1}$.  This implies that the space of cyclic $0$-cocycles is isomorphic to $\R$. As there are no  cyclic $-1$-cochains, we conclude that 
$$HC^0_\lambda(\R^{0|1},\partial_\zx)_0 =0\,, \qquad HC^0_\lambda(\R^{0|1},\partial_\zx)_1 \cong \R\,.$$  
\end{example}
\begin{example}
Consider $(\Pi \sT S, \rmd)$ which employ coordinates $(\vartheta, \rmd\vartheta)$, with $\vartheta \in [ 0, 2 \pi)$. The de Rham differential in these coordinates is $\rmd = \rmd \vartheta \partial_\vartheta$. The standard cohomology, once form degree has been included, corresponds to the de Rham cohomology. Recall, $H^0_{\mathrm{dR}}(S) \cong \R$ and $H^1_{\mathrm{dR}}(S) \cong \R$, while all higher de Rham cohomology groups are trivial.  Note $\Pi \sT S$ is unimodular, superorientable and compact. Functions on $\Pi \sT S$ are of the form $\alpha = \alpha_0(\vartheta) + \rmd \vartheta \, \alpha_1(\vartheta)$. Any cyclic $0$-cochain decomposes into even and odd parts (here defined with respect to the original grading) as
 $$\tau(\alpha_0 + \rmd \vartheta \, \alpha_1) = \tau_0(\alpha_0) + \tau_1(\rmd \vartheta \, \alpha_1)\,.$$
Enforcing the condition $\tau \circ \rmd =0$ implies that $\tau(\rmd \alpha) = \tau_1(\rmd \vartheta \, \partial_\vartheta\alpha_0)=0$. Thus, $\tau_1$ must annihilate  exact forms, implying  $\tau_1(\rmd \vartheta \,\alpha_1) = c \, \int_{\Pi \sT S} D[\vartheta , \rmd \vartheta] \,\rmd \vartheta \, \alpha_1(\vartheta)$, for some $c \in \R$.  There is no constraint on $\tau_0$. Then we have 
$$\tau(\alpha_0 + \rmd \vartheta\, \alpha_1) = \tau_0(\alpha_0(\vartheta)) +  c \, \int_{\Pi \sT S} D[\vartheta , \rmd \vartheta] \, \rmd \vartheta \,\alpha_1(\vartheta) \,,$$
as the general form of a cyclic $0$-cocycle.  We then conclude that 
$$HC^0_\lambda(\Pi \sT S,\rmd) \cong C^0_\lambda(\Pi \sT S, d)_0 \oplus \R\,.$$ 
\end{example}
\begin{example}
Continuing Example \ref{exa:HigherPoisson}, let $(M, \mathcal{P})$ be a higher Poisson manifold. Then a cyclic $0$-cocycle satisfies $\tau \big( (\mathcal{P},  \mathcal{X})\big)=0$, for all $\mathcal{X} \in C^\infty(\Pi \sT^* M)$. Any odd Poisson trace, i.e, a functional $\tau : C^\infty(\Pi \sT^* M) \rightarrow \R$ that satisfies $\tau \big( (\mathcal{X}, \mathcal{Y})\big) =0$ for all $\mathcal{X}, \mathcal{Y}\in C^\infty(\Pi \sT^* M)$, is a cyclic $0$-cocycle. Moreover, defining the Hamiltonian vector field $X_{\mathcal{X}}:= (-1)^{\widetilde{\mathcal{X}}}\, (\mathcal{X}, -)$, we observe that any odd Poisson trace satisfies $\tau \circ X_{\mathcal{X}} =0$.
\end{example}
\begin{example}
Consider $(S^{1|1}, Q)$, where $Q$ is a non-singular homological vector field. Then via Shander--Vaintrob \cite{Shander:1980,Vaintrob:1996} we know there is a Q-diffeomorphism $S^{1|1} \rightarrow S^1 \times \R^{0|1}$ rendering $Q= \partial_\zx$, where $\zx$ is the global coordinate on $\R^{0|1}$.
Any function $C^\infty(S^{1|1})$ is of the form $f = f_0 + f_1 \, \zx$.  We know that   any derived cyclic $0$-cocycle must satisfy $\tau \circ Q =0$. Thus, $\tau(Qf) = \tau(f_1) =0$, and $\tau$ therefore must annihilate functions in $C^\infty(S^1)$. This implies that $\tau$ is determined by its action on the odd sector $C^\infty(S^1)\cdot \zx$. As there are no further conditions, we conclude that 
$$HC^0_\lambda(S^{1|1}, Q)_0 =0\, , \qquad HC^0_\lambda(S^{1|1}, Q)_1 \cong \big( C^\infty(S^1)\big)^*\,.$$
\end{example}
\begin{example}
Consider $(\R^{n|m}, Q)$ where $Q$ is taken to be non-singular. Then via Shander--Vaintrob  \cite{Shander:1980,Vaintrob:1996}, we may equip the supermanifold with (global) coordinates $(x^a, \zx, \theta^\alpha)$, with $a= 1, \cdots, n$ and $\alpha = 1, \cdots m-1$, and $Q = \partial_\zx$. We write $f = f_0(x, \theta) + \zx\, f_1(x, \theta) \in C^\infty(\R^{n|m})$. Then $Qf = f_1(x, \theta)$.  We know that   any derived cyclic $0$-cocycle must satisfy $\tau \circ Q =0$. Thus, $\tau(f_1(x, \theta))=0$.  We then define $\tau_1(g(x, \theta)) := \tau (\zx \, g(x,\theta))$ and thus 
\begin{align*}
& HC^0_\lambda(\R^{n|m}, Q)_0 \cong \InHom_{\catname{NFSVec}}\big( C^\infty(\R^{n| m-1}), \R\big)_1  \, , \\
& HC^0_\lambda(\R^{n|m}, Q)_1 \cong \InHom_{\catname{NFSVec}}\big( C^\infty(\R^{n|m-1}), \R\big)_0\,.
\end{align*}
\end{example}
The above examples, while modest in scope, clearly demonstrate that derived cyclic cohomology is quite independent from the standard cohomology of a Q-manifold. \par 
We expect, in general, calculating the derived cyclic cohomology and explicitly presenting cyclic $n$-cocycles for  Q-manifolds without extra structure to be very difficult.   It is important to note that the derived associative algebra is non-unital and so many of the standard methods of calculating cyclic cohomology, for example Connes' SBI sequence, will not directly apply to derived cyclic cohomology. We will comment on this a little further in our concluding remarks. In Subsection \ref{subsec:OddLine} we explore the derived cyclic cohomology of the odd line extending Example \ref{exa:OddLine}.  We will present derived cyclic $n$-cocycles for regular foliations via their foliation Lie algebroid in Subsection \ref{subsec:RegFoli}. Further examples are left open for future work. 
%
%
\subsection{Derived Cyclic Cohomology of the Odd Line}\label{subsec:OddLine}
Consider the odd line $M = \R^{0|1}$, which we equip with a global coordinate $\zx$, and a canonical homological vector field given by $Q = \partial_\zx$. The derived product of two functions $f = a + b\, \zx$ and $g = c + d \, \zx$ ($a,b,c$ and $d \in \R$) is  $f\star g = - bc - bd \, \zx$. Defining the linear map $\eta(a + b \, \zx) = b$, we can write the derived product as $f\star g =  - \eta(f)g$.  As a super vector space, $C^\infty(\R^{0|1}) = \R[\zx]$ has a canonical basis $\{1, \zx \}$. In turn, this implies we can write the Hochschild coboundary operator, setting each $f \in \{ 1 , \zx\}$ as 
\begin{align}\label{eqn:HochCobOddLine}
(b \phi^n)(f_0, \cdots , f_{n+1}) &= \sum_{i=0}^n (-1)^{i+1}\, \eta(f_i)\, \phi^n(f_0, \cdots, \hat{f}_i, \cdots , f_{n+1})\\ \nonumber 
&+ (-1)^{n +\overline{f}_{n+1}(\overline{f}_0 + \cdots + \overline{f}_n)}\, \eta(f_{n+1})\, \phi^n(f_0, \cdots , f_n)\\ \nonumber
& = \sum_{i=0}^{n+1} (-1)^{i+1}\, \eta(f_i)\, \phi^n(f_0, \cdots, \hat{f}_i, \cdots , f_{n+1})\,.
\end{align}  
\begin{theorem}\label{thm:DerCycCoOddLine}
Let the odd line $M = \R^{0|1}$ be equipped with a global coordinate $\zx$, and the canonical homological vector field given by $Q = \partial_\zx$. Then 
$$
HC^n_\lambda(\R^{0|1}, \partial_\zx) = \begin{cases} 
\R^{0|1} = 0 \oplus \R \,, & \text{if~} n \text{~is even} \\ 
0 = 0\oplus 0 \,, & \text{if~} n \text{~is odd} 
\end{cases}\,.
$$
\end{theorem}
\begin{proof}
Let us set $f \in \{ 1, \zx \}$. We decompose derived $n$-cochains $C^n_\lambda := C^n_\lambda(\R^{0|1},\partial_\zx)$, via the number of $1$'s that appear in the input tuple and write
$$C^n_\lambda = \bigoplus_{m=0}^{n+1} C^{n,m}_\lambda\,.$$
As $\eta(f) =1$ only when $f = \zx$, contracting an argument via $b$ (see \eqref{eqn:HochCobOddLine}) preserves $m$. Thus, we may restrict the Hochschild coboundary operator $b : C^{n,m}_\lambda \rightarrow C^{n+1,m}_\lambda$.
\begin{description}
\item[$\boldsymbol{m\geq 1}$] We will show that this sector does not contribute to the derived cyclic cohomology.  For any and all $\phi \in C^{n,m}_\lambda$ ($m \geq 1$) and $b \phi =0$, we claim that the operator 
\begin{equation}\label{eqn:ContHomOp}
(s_m \phi)(f_0, \cdots , f_{n-1}) =  \frac{1}{m} \sum_{i=0}^{n-1} (-1)^{i +1}\, \phi(f_0, \cdots , f_{i-1}, \zx, f_i, \cdots , f_{n-1})\,,
\end{equation}
satisfies the following:\\
\begin{enumerate}[i)]
\setlength{\itemsep}{6pt}
\item $\lambda (s_m \phi) = s_m \phi\,$;
\item $\phi = b(s_m \phi)\,$.
\end{enumerate}
Together, establishing these will complete this part of the proof. \\
\begin{enumerate}[i)]
\setlength{\itemsep}{6pt}
\item We need to demonstrate that $(-1)^\epsilon \, (s_m \phi)(f_{n-1}, f_0, f_{n-2}) = (s_m\phi) (f_0, \cdots , f_{n-1})$, where the sign factor is $\epsilon =  n-1 + \overline{f}_{n-1}(\overline{f}_0 + \cdots + \overline{f}_{n-2})$. 
\begin{align*}
(-1)^\epsilon \, (s_m \phi)(f_{n-1}, f_0,\cdots, f_{n-2}) & = (-1)^\epsilon \, \frac{1}{m} \left(\vphantom{\sum_{j=0}^{n-2}}- \phi(\zx, f_{n-1}, f_0, \cdots f_{n-2}) \right.\\ & \left.+ \sum_{j=0}^{n-2} (-1)^{j+2} \, \phi(f_{n-1}, f_0 , \cdots , f_{j-1}, \zx, f_j, \cdots , f_{n-2}) \right)\\
\intertext{Using the shifted cyclic property of $\phi$, we can move $f_0$ to the first entry. Taking care with the sign factors, we arrive at}
&=  \frac{1}{m} \left(\vphantom{\sum_{j=0}^{n-2}} (-1)^n \, \phi( f_0, \cdots,  f_{n-2},\zx, f_{n-1} ) \right.\\ & \left.+ \sum_{j=0}^{n-2} (-1)^{j+2} \, \phi(f_0 , \cdots , f_{j-1}, \zx, f_j, \cdots , f_{n-1}) \right)\\
&= \frac{1}{m} \sum_{j=0}^{n-1} (-1)^{j+1} \phi(f_0, \cdots , f_{j-1}, \zx, f_j \cdots , f_{n-1})\\
&= (s_m)\phi(f_0, \cdots , f_{n-1})\,.
\end{align*}
\item This is equivalent to $s_m$ being derived cyclic cochain contraction operator; that is, $b \circ s_m +  s_m \circ b = \textrm{id}$. Directly from \eqref{eqn:HochCobOddLine} and \eqref{eqn:ContHomOp}, decomposing the expression to if the insertion of $\zx$ is before of after the removal of $f_i$, we have
\begin{align} \label{eqn:bsm}
b(s_m\phi)(f_0, \dots, f_n) &= \frac{1}{m} \sum_{i=0}^n \sum_{k=0}^{i-1} (-1)^{i+k} \eta(f_i) \phi(f_0, \dots, \underset{\text{pos~} k}{\zx}, \dots, \hat{f}_i, \dots, f_n)\\ \nonumber
&+ \frac{1}{m} \sum_{i=0}^n \sum_{k=i+1}^n (-1)^{i+k+1} \eta(f_i) \phi(f_0, \dots, \hat{f}_i, \dots, \underset{\text{pos~} k}{\zx}, \dots, f_n)\,.
\end{align}
In the above  $\zx$ in in the $k^\text{th}$ position in the sums.\\ 
Next we evaluate
$$s_m(b\phi)(f_0, \dots, f_n) = \frac{1}{m} \sum_{p=0}^n (-1)^{p+1} (b\phi)(f_0, \dots, f_{p-1}, \zx, f_p, \dots, f_n)\,.$$
To do this, we set $(g_0, \dots, g_{n+1}) := (f_0, \dots, f_{p-1}, \zx, f_p, \dots, f_n)$ and understand $\zx$ sitting on the $p^\text{th}$ position with respect to the new labelling. We now expand $b\phi$ as before, taking care with which element $g_i$ is removed. There are three distinct cases:
\begin{enumerate}[1.]
\item $j < p$: The removed element is $f_j$. The inserted element $\zx$ shifts to index $k = p-1$ in the arguments of $\phi$. The overall coefficient then becomes $\frac{1}{m}(-1)^{p+1}(-1)^{j+1} = \frac{1}{m}(-1)^{j+k+1}$.
\item $j = p$: The removed element is the inserted $\zx$. This leaves the sequence $(f_0, \dots, f_n)$ unchanged. Given $\eta(\zx)=1$, the coefficient is $\frac{1}{m}(-1)^{p+1}(-1)^{p+1} = \frac{1}{m}$.
\item $j > p$: The removed element is $f_i$, where $i = j-1 \geq p$. The inserted element $\zx$ remains at index $k = p$. The overall coefficient becomes $\frac{1}{m}(-1)^{p+1}(-1)^{i+2} = \frac{1}{m}(-1)^{i+k+1}$.
\end{enumerate}
Combining these three cases, we arrive at 
\begin{align}\label{eqn:smb}
s_m(b\phi)(f_0, \dots, f_n) &= \frac{1}{m} \sum_{k=0}^{n-1} \sum_{i=0}^k (-1)^{i+k+1} \eta(f_i) \phi(f_0, \dots, \hat{f}_i, \dots, \underset{\text{pos~} k}{\zx}, \dots, f_n)\\ \nonumber 
& + \frac{n+1}{m} \phi(f_0, \dots, f_n)\\ \nonumber
& + \frac{1}{m} \sum_{k=0}^n \sum_{i=k}^n (-1)^{i+k+1} \eta(f_i) \phi(f_0, \dots, \underset{\text{pos~} k}{\zx}, \dots, \hat{f}_i, \dots, f_n)\,.
\end{align}
Comparing \eqref{eqn:bsm} and \eqref{eqn:smb}, observe the relative minus signs between many (but not all) of the similar terms.   Examining $b\circ m_s + m_s \circ b$, we see that 
\begin{itemize}
\item The first sum of $s_m(b\phi)$ (where $i \leq k$)  cancels the second sum of $b(s_m\phi)$ (where $i < k$, bounded differently but referencing identical terms with opposite signs).
\item The $i > k$ components of the third sum of $s_m(b\phi)$  cancel the first sum of $b(s_m\phi)$.
\end{itemize}
The remaining terms produce
\begin{align}\label{eqn:bsm+smb}
(b s_m + s_m b)\phi(f_0, \dots, f_n) &= \frac{n+1}{m} \phi(f_0, \dots, f_n)\\ \nonumber
 &- \frac{1}{m} \sum_{k=0}^n \eta(f_k) \phi(f_0, \dots, f_{k-1}, \zx, f_{k+1}, \dots, f_n)\,.
\end{align}
Note that each term of the sum in \eqref{eqn:bsm+smb} is only non-zero when $f_k = \zx$. As $\eta(\zx)=1$, the summation counts the number of occurrences of $\zx$ in the argument sequence $(f_0, \dots, f_n)$ and scales $\phi(f_0, \dots, f_n)$ by that integer. By the definition of $C^{n,m}_\lambda$,  cochains act on arguments with  $m\geq 1$ occurrences of $1$. Since the sequence contains $n+1$ total arguments, it must contain exactly $(n+1-m)$ instances of $\zx$.\par 
Using this deduction in \eqref{eqn:bsm+smb}, we have
$$(b s_m + s_m b)\phi = \frac{n+1}{m}\phi - \frac{n+1-m}{m}\phi = \frac{m}{m}\phi = \phi\,.$$ 
In other words, $b s_m + s_m b = \mathrm{id}$ and so a derived cyclic cochain contraction operator.  
\end{enumerate}
%
%
\item [$\boldsymbol{m =0}$] The shifted cyclic condition means $\phi^n(\zx, \cdots, \zx)= (-1)^n \, \phi^n(\zx, \cdots ,\zx)$. Thus, if $n$ is odd, then $\phi^n(\zx, \cdots, \zx)=0$.  For $n$ even, applying the Hochschild coboundary operator, we obtain 
$$(b\phi^n)(\zx, \cdots, \zx) = \left(\sum_{i=0}^n (-1)^{i+1} + (-1)^n  \right)\phi^n(\zx, \cdots, \zx) =0\,.$$
\end{description}
Consequently, only the $m=0$ $n$-cochains with $n$ even contribute to the cohomology. Moreover, every such cochain is a cocycle and cannot have a $b$-primitive.  Taking care with the Grassmann parity, we have shown that 
$$
HC^n_\lambda(\R^{0|1}, \partial_\zx) = \begin{cases} 
0\oplus\R, & \text{if~} n \text{~is even} \\ 
0\oplus 0, & \text{if~} n \text{~is odd} 
\end{cases}\,.
$$
\end{proof}
\begin{corollary}\label{cor:CycSandiso}
Let $HC^n_\lambda(\R)$ be the (standard) cyclic cohomology of the real numbers. There are isomorphisms  $HC^n_\lambda(\R^{0|1}, \partial_\zx)_1 \cong HC^n_\lambda(\R)$ for $n=0, 1,2, \cdots$.
\end{corollary}
We stress that the results in this subsection do not depend on the choice of odd coordinate. Specifically, any coordinate change is of the form $\zx' = a\, \zx$, with $a \in \R_*$. Then $\varphi^* \zx' = \zx$  (which is just the rescaling $\varphi^*\zx = a^{-1}\, \zx$) is  a Q-diffeomorphism and via  Theorem \ref{thm:DerCycCoFunct} we know the derived cyclic cohomology is unchanged.  \par 
Recall that the ($\Z_2$-graded) cyclic cohomology of $C^\infty(\R^{0|1}) = \R[\zx]$ is 
$$
HC^n_\lambda(\R[\zx]) = \begin{cases} 
\R\oplus \R, & \text{if~} n \text{~is even} \\ 
0\oplus \R, & \text{if~} n \text{~is odd} 
\end{cases}\,.
$$
Thus, while clearly related, the derived and standard cyclic cohomologies are not isomorphic. Moreover, $H^\bullet_{\textrm{std}}(\R^{0|1}, \partial_\zx) =0$, i.e., completely trivial.  
%
%
\subsection{Derived Cyclic Cocycles on Double Q-manifolds}\label{subsec:DoubleShiftSycCoh}
Given a double Q-manifold $(M, Q_1, Q_2)$, we define a pencil of homological vector fields as $Q_t := (1-t)\, Q_1 + t \, Q_2$, with $t \in (0,1)$. This pencil interpolates between the two initial homological vector fields.  We then have a pencil of derived associative algebras with the derived product being defined as 
\begin{equation}
f \star_t g  := (-1)^{\widetilde{f}} \, Q_t(f) \, g\,
\end{equation}
with $f,g \in C^\infty(M)$.  We then define the obvious pencil of Hochschild coboundary operators $b_t$ and derived cyclic cohomology groups.  Then, the derived cyclic $n$-cocycles, from Proposition \ref{prop:CochainsInv}, satisfies
$$\phi^n \circ (Q_t \otimes Q_t \otimes \cdots \otimes Q_t)= 0\,,$$
for all $t \in (0,1)$. To analyse this condition, we define 
\begin{equation} 
S_k := \sum_{\sigma \in \mathcal{S}_{n+1,k}} Q_{\sigma(1)} \otimes Q_{\sigma(2)}\otimes \cdots \otimes Q_{\sigma(n+1)}\,,
\end{equation}
where $\mathcal{S}_{n+1,k}$ is the set of sequences of length $n+1$ consisting of $k$ copies of the index $2$ and $n+1-k$ copies of the index $1$. Using the binomial expansion and relabelling, we obtain 
\begin{align} \label{eqn:DoubleInvCoh}
\phi^n \circ (Q_t \otimes Q_t \otimes \cdots \otimes Q_t) & = \phi^n \circ \left( \sum_{k=0}^{n+1}(1-t)^{n+1-k}\, t^k \, S_k\right)\\ \nonumber 
& = \phi^n \circ \left(\sum_{m=0}^{n+1} t^m \, \left( \sum_{k=0}^m  (-1)^{m-k} \, \begin{pmatrix}
n+1 -k \\
n+1 -m 
\end{pmatrix}\, S_k \right) \right) =0\,.
\end{align}
\begin{proposition}\label{prop:DoubleCochainsInv}
Let $(M, Q_1, Q_2)$ be a double Q-manifold, then all $b_t$-closed cochains $\phi^n \in C^n(M,Q_t)$  satisfy
$$\phi^n \circ \left (\sum_{k=0}^m  (-1)^{m-k} \, \begin{pmatrix}
n+1 -k \\
n+1 -m 
\end{pmatrix} \, S_k\right) =0\,,$$
for each $m \in \{ 0,1,2, \cdots , n+1\}$.
\end{proposition}
\begin{proof}
The equation \eqref{eqn:DoubleInvCoh} holds order by order in $t$. Thus, isolating these terms establishes the proposition. 
\end{proof}
For instance, one can quickly deduce that for a derived cyclic $0$-cocycle $\tau$
$$\tau \circ Q_1 =0\,, \qquad \tau \circ Q_2 =0\,.$$
For derived cyclic $1$-cocycles we deduce that 
$$\phi^1 \circ(Q_1 \otimes Q_1)=0\,, \qquad \phi^1 \circ(Q_2 \otimes Q_2)=0\,, \qquad \phi^1 \circ (Q_1 \otimes Q_2)+ \phi^1 \circ (Q_2 \otimes Q_1)=0\,.$$
Similarly, for derived cyclic $2$-cocycles we deduce that  
\begin{align*}
&\phi^2 \circ (Q_1 \otimes Q_1 \otimes Q_1) =0\,,
 \qquad \phi^2 \circ (Q_2 \otimes Q_2 \otimes Q_2) =0\,,\\
&\phi^2 \circ (Q_2 \otimes Q_1 \otimes Q_1) + \phi^2 \circ (Q_1 \otimes Q_2 \otimes Q_1) +\phi^2 \circ (Q_1 \otimes Q_1 \otimes Q_2) =0\,,\\
&\phi^2 \circ (Q_1 \otimes Q_2 \otimes Q_2) + \phi^2 \circ (Q_2 \otimes Q_1 \otimes Q_2) +\phi^2 \circ (Q_2 \otimes Q_2 \otimes Q_1) =0\,.
\end{align*}
We will denote the derived cyclic cocycles as $C^n_\lambda(M, Q_1, Q_2)$.
\begin{example} 
Let $(\mathfrak{g}, [-,-])$ be a pure even finite-dimensional Lie algebra.  The supermanifold we consider is $M := \Pi \mathfrak{g}\oplus \Pi \mathfrak{g} \oplus \mathfrak{g} \cong \R^{r|2r}$, where $\dim \mathfrak{g} = r$.  We equip $M$ with (global) coordinates $(\eta^a, \bar{\eta}^b, b^c)$, which in physical language are Faddeev--Popov ghosts, Faddeev--Popov anti-ghosts, and Nakanishi--Lautrup auxiliary fields, respectively; recall we have suppressed ghost number. We can then identify $C^\infty(M) \cong \Lambda^\bullet \mathfrak{g}^* \otimes \Lambda^\bullet \mathfrak{g}^* \otimes C^\infty( \mathfrak{g}^*) $. The BRST differential and the anti-BRST differential 
\begin{align*}
 &s = Q_1 :=  \frac{1}{2}\eta^a \eta^b Q_{ba}^c \frac{\partial}{\partial \eta^c} + b^a \frac{\partial}{\partial \bar{\eta}^a}\,,  \\ 
 & \bar{s} = Q_2 := \frac{1}{2} \bar{\eta}^a \bar{\eta}^b Q_{ba}^c \frac{\partial}{\partial \bar{\eta}^c} - (b^c - \eta^a \bar{\eta}^b Q_{ba}^c)\frac{\partial}{\partial \eta^c} - b^a \bar{\eta}^bQ_{ba}^c \frac{\partial}{\partial b^c}\,,
\end{align*} 
define a double Q-manifold $(M, s, \bar{s})$. Here $Q_{ba}^c$ are the structure constants of the Lie algebra.  Note that the first terms of $s$ and $\bar{s}$ are just the Chevalley--Eilenberg differential of the Lie algebra $\mathfrak{g}$. Then via Proposition \ref{prop:ComParAlgs}, we have a pair of compatible derived algebras on $C^\infty(M)$.  
\begin{align*}
& \eta^a \star_s \eta^b = -\frac{1}{2}  \eta^c \eta^d Q^a_{dc}\eta^b\,, && \eta^a \star_s \bar{\eta}^b = -\frac{1}{2}  \eta^c \eta^d Q^a_{dc}\bar{\eta}^b\,, && \eta^a \star_s b^b = -\frac{1}{2}  \eta^c \eta^d Q^a_{dc}b^b\,,\\
&\bar{\eta}^a \star_s \eta^b = -b^a \eta^b\,, && \bar{\eta}^a \star_s \bar{\eta}^b = -b^a \bar{\eta}^b\,, && \bar{\eta}^a \star_s b^b = -b^a b^b\,,\\
& \eta^a \star_{\bar{s}} \eta^b = \left(b^a - \eta^c \bar{\eta}^d Q^a_{dc}\right) \eta^b\,, && \eta^a \star_{\bar{s}} \bar{\eta}^b = \left(b^a - \eta^c \bar{\eta}^d Q^a_{dc}\right) \bar{\eta}^b\,, && \eta^a \star_{\bar{s}} b^b = \left(b^a - \eta^c \bar{\eta}^d Q^a_{dc}\right) b^b\,,\\
& \bar{\eta}^a \star_{\bar{s}} \eta^b = -\frac{1}{2} \bar{\eta}^c \bar{\eta}^d Q^a_{dc} \eta^b\,, && \bar{\eta}^a \star_{\bar{s}} \bar{\eta}^b = -\frac{1}{2} \bar{\eta}^c \bar{\eta}^d Q^a_{dc} \bar{\eta}^b\,, && \bar{\eta}^a \star_{\bar{s}} b^b = -\frac{1}{2} \bar{\eta}^c \bar{\eta}^d Q^a_{dc} b^b\,,\\
& b^a \star_{\bar{s}} \eta^b = -b^c \bar{\eta}^d Q^a_{dc} \eta^b\,, && b^a \star_{\bar{s}} \bar{\eta}^b = -b^c \bar{\eta}^d Q^a_{dc} \bar{\eta}^b\,, && b^a \star_{\bar{s}} b^b = -b^c \bar{\eta}^d Q^a_{dc} b^b\,.
\end{align*}
All derived cyclic $0$-cocycles $\tau : C^\infty(M) \rightarrow \R$ are both BRST and anti-BRST invariant, i.e., $\tau \circ s =0$ and $\tau \circ \bar{s}=0$.    Any derived cyclic $1$-cocycle satisfies
$$\phi^1(s f, s g) =0\,, \qquad \phi^1(\bar{s}f , \bar{s}g)=0\, , \qquad \phi^1(sf , \bar{s}g) + \phi^1(\bar{s}f, s g)=0\,,$$ 
for all $f,g \in \Lambda^\bullet \mathfrak{g}^* \otimes \Lambda^\bullet \mathfrak{g}^* \otimes C^\infty( \mathfrak{g}^*)$.\par 
As a concrete case, consider $\mathfrak{g} = \mathfrak{u}(1) \cong \R$.  In this case we have $s = b \partial_{\bar{\eta}}$ and $\bar{s} = - b \partial_\eta$. Expanding a function as $f = f_0(b) + f_1(b)\eta  + f_2(b) \bar{\eta} + f_3(b) \eta \bar{\eta}$.  We then have 
$$sf  = b(f_2(b) - f_3(b)\eta )\,, \qquad \bar{s} f  = -b(f_1(b)+ f_3(b)\bar{\eta})\,.$$ 
Then, as all derived cyclic $0$-cocycles are both BRST and anti-BRST invariant, we have a series of conditions 
\begin{multicols}{2}
\begin{enumerate}[i)]
\item $\tau(b f_2(b)) = 0$,
\item $\tau(b f_1(b)) = 0$,
\item $\tau(b f_3(b)\eta) = 0$,
\item $\tau(b f_3(b)\bar{\eta}) = 0$.
\end{enumerate}
\end{multicols}
These conditions place constraints on the first three terms of the evaluation of $\tau$. Specifically, they are determined by the value of $f$ at $b=0$.  Thus, we have 
$$\tau(f_0(b)) = a \, f_0(0)\,, \qquad \tau(f_1(b)\eta) = b \, f_1(0)\,, \qquad \tau(f_2(b)\bar{\eta}) =c\, f_2(0)\,,$$
with $a,b$ and $c \in \R$. Finally, $\tau(f_3(b) \eta \bar{\eta})$ is unconstrained, as such a term can never be generated by $s$ or $\bar{s}$. \par 
Thus, we have deduced that 
$$C^0_\lambda(\Pi \R \oplus \Pi \R \oplus \R, s, \bar{s}) \cong \big(\R \oplus \Hom_{\catname{Vec}}(C^\infty(\R), \R) \big)\oplus \R^2\,.$$
\end{example}
%
%
\subsection{Applications to Lie Algebras: The Relation with Lie Algebra Homology}
In this subsection, we will consider  a pure even finite dimensional Lie algebra $(\mathfrak{g}, [-,-])$.  For details of the homology and cohomology of Lie algebras, the reader may consult Weibel \cite[Chapter 7]{Weibel:1994}. Let us choose a basis $\{e_a \}$, such that $[e_a, e_b] = Q_{ab}^{c}e_c$. Shifting the parity of the vector space (now considered as a linear supermanifold), we obtain a Q-manifold $(\Pi \mathfrak{g}, \rmd_{CE})$, where, in coordinates,
\begin{equation}
\rmd_{CE} = \frac{1}{2} \eta^b \eta^a Q_{ab}^c \frac{\partial}{\partial \eta^c}\,,
\end{equation}
is recognised as the Chevalley--Eilenberg differential.  Note $C^\infty(\Pi \mathfrak{g}) \cong \Lambda^\bullet \mathfrak{g}^*$.  We remark that $\rmd_{CE}^2 =0$ is equivalent to the Jacobi identity for the Lie bracket. Borrowing language from physics, the coordinates $\eta^a$ are Faddeev--Popov ghosts. \par 
Clearly, we have the derived associative algebra $(C^\infty(\Pi \mathfrak{g}), \star)$ associated with any Lie algebra.  The derived $0$-cochains are
$$C^0(\Pi \mathfrak{g},\rmd_{CE}):= \InLin\big(\Lambda^\bullet\mathfrak{g}^*, \R\big) = \big(\Lambda^\bullet \mathfrak{g}^* \big)^* \cong \Lambda^\bullet \mathfrak{g}\,.$$
Note that the algebra is finite dimensional, and so the completed tensor product coincides with the algebraic tensor product, meaning we lose no generality considering linear maps  between the finite-dimensional super vector spaces. Via  Lemma \ref{lem:0Cocycles}, all derived cyclic $0$-cocycles satisfy $\tau \circ \rmd_{CE} =0$. The canonical pairing $\langle -,- \rangle$ between $\Lambda^\bullet \mathfrak{g}$ and $\Lambda^\bullet \mathfrak{g}^*$ allows us to define the Chevalley--Eilenberg boundary operator $\partial_{CE}$ via
\begin{equation}
\langle \tau, \rmd_{CE} \alpha \rangle = \langle \partial_{CE}\tau, \alpha \rangle\,.
\end{equation}
Then using the canonical isomorphism $\big( \Lambda^\bullet \mathfrak{g}^*\big)^* \cong \Lambda^\bullet \mathfrak{g}$, we identify derived cyclic $0$-cocycles  with Lie algebra (homology) cycles.  More formally, we have established the following.
\begin{proposition}\label{prop:LieAlgH0}
Let $(\mathfrak{g}, [-,-])$ be a (pure even) finite dimensional Lie algebra of dimension $n$. Then 
$$HC^0_\lambda(\Pi \mathfrak{g}, \rmd_{CE}) \cong \mathcal{Z}_\bullet(\mathfrak{g})= \bigoplus_{k=0}^n \mathcal{Z}_k(\mathfrak{g})\,,$$
where $\mathcal{Z}_k(\mathfrak{g})$ are the Lie algebra cycle groups, and we assign Grassmann parity of $k$ (mod $2$) to elements of the Lie algebra $k$-cycle group. 
\end{proposition}
\begin{example}
For any abelian Lie algebra $(\mathfrak{a}, [-,-])$ of dimension $n$, as the Chevalley--Eilenberg boundary operator is the zero map
$$HC^0_\lambda(\Pi \mathfrak{a}, \rmd_{CE})\cong \mathcal{Z}_\bullet(\mathfrak{a}) = \bigoplus_{k=0}^n \Lambda^k \mathfrak{a} \cong \R^{2^{n-1} | 2^{n-1}}\,.$$
\end{example}
\begin{example}
Consider the $3$-dimensional Heisenberg algebra $\mathfrak{h}$ given by generators $\{ x,y,z \}$ and non-vanishing Lie bracket $[x, y] = z$. Note $z$ is central.  Then $\mathcal{Z}_\bullet(\mathfrak{h})= \mathcal{Z}_0(\mathfrak{h}) \bigoplus \mathcal{Z}_1(\mathfrak{h})\bigoplus \mathcal{Z}_2(\mathfrak{h}) \bigoplus \mathcal{Z}_3(\mathfrak{h})$. Observe that $\mathcal{Z}_0(\mathfrak{h}) = \R$, and so $\mathcal{Z}_1(\mathfrak{h}) = \mathfrak{h}$.  From the algebra we see that $\partial_{CE}(x \wedge y) = - z$ and $\partial_{CE}(x \wedge z) = \partial_{CE}(y \wedge z) =0$. Thus, $\mathcal{Z}_2(\mathfrak{h}) = \Span(x \wedge z, y \wedge z)$. Finally, $\partial_{CE}(x \wedge y \wedge z) = - z \wedge z =0$, thus $\mathcal{Z}_3(\mathfrak{h}) = \Span(x \wedge y \wedge z)$.  In total we have 
$$HC^0_\lambda(\Pi \mathfrak{h}, \rmd_{CE})\cong  \R \bigoplus \Span(x,y,z) \bigoplus \Span(x \wedge z,y \wedge z ) \bigoplus \Span(x \wedge y \wedge z)  \cong \R^{3|4}\,. $$
\end{example}
While $C^n_\lambda(\Pi \mathfrak{g}, \rmd_{CE}) \subset \big(\Lambda^\bullet \mathfrak{g} \big)^{\otimes^{n+1}} \cong \InLin\big((\Lambda^\bullet\mathfrak{g}^*)^{\otimes^{n+1}}, \R\big)$, we do not have a simple generalisation of Proposition \ref{prop:LieAlgH0} for $n>0$ and non-abelian Lie algebras.  In particular, for the $n>0$ case, we have $b$-exact cochains to contend with when calculating the derived cyclic cohomology, and this prevents a simple decomposition relating it to the Lie algebra homology. 
%
%
\subsection{Applications to Regular Foliations: Generalised Currents}\label{subsec:RegFoli}
Let $N$ be a pure even manifold and let $A \subset \sT N$ be an integrable subbundle defining a regular foliation $\mathcal{F} = \{L_i \}_{i \in \mathcal{I}}$ of $N$. This defines the foliation Lie algebroid of the regular foliation. Sections of $A$ are vector fields tangent to $\mathcal{F}$. The Lie bracket is the restriction of the standard Lie bracket of vector fields, and the anchor is the inclusion map (and so it is injective). Moreover, the orbit foliation of the anchor is precisely the foliation $\mathcal{F}$. Conversely, any Lie algebroid with an injective anchor map is isomorphic to a foliation Lie algebroid. \par 
The corresponding Q-manifold picture is as follows. The supermanifold is $\Pi A$ (the parity reversed vector bundle), and so $C^\infty(\Pi A)$ is identified with leaf-wise (pseudo-)differential forms $\Omega^\bullet(\mathcal{F})$. The homological vector field is the restriction of the de Rham differential to the leaves of the foliation, i.e., $Q = \rmd_{\mathcal{F}}$.  The derived product is then $\alpha \star \beta := (-1)^{\widetilde{\alpha}} \, \rmd_{\mathcal{F}} \alpha \wedge \beta$, with $\alpha, \beta \in \Omega^\bullet(\mathcal{F})$. The modular class of a foliation Lie algebroid (see Subsection \ref{subsec:ModClass}) is the obstruction to the existence of a global, strictly positive, nowhere-vanishing transverse invariant volume. In classical language, the modular class of a foliation Lie algebroid is known as the Reeb class.\par 
Moving to derived cyclic cohomology of the foliation, the derived $n$-cochains are elements of 
$$C^n(\Pi A, \rmd_{\mathcal{F}}):= \InHom_{\catname{NFSVec}}\big( (\Omega^\bullet(\mathcal{F}))^{\widehat{\otimes}^{n+1}}, \R\big)\,.$$
Some comments:
\begin{enumerate}[i)]
\item Derived $0$-cochains are generalised foliation currents; we recover classical foliation currents when restricting to leaf-wise top-forms with compact support.  
\item Derived cyclic $0$-cocycles, via Lemma \ref{lem:0Cocycles} (or Proposition \ref{prop:CochainsInv}), satisfy $\tau \circ \rmd_{\mathcal{F}}=0$. A foliation current that vanishes on leaf-wise exact forms is a closed foliation current (or a holonomy-invariant transverse current), see  Ruelle \& Sullivan \cite{Ruelle:1975}. Thus, derived cyclic $0$-cocycles are generalised closed foliation currents. 
\item Derived cyclic $n$-cocycles ($n \geq 1$) are interpreted as generalised closed foliation multiform currents. Their precise interpretation in terms of classical foliation theory is currently unclear. As far as the author is aware, multilinear functionals of leaf-wise differential forms have not been considered before. 
\end{enumerate}
\begin{example}\label{exa:T2}
Consider the torus $T^2 = S^1 \times S^1$ which we equip with angular coordinates $(\vartheta, \varphi)\in [0, 2\, \pi) \times [0, 2 \, \pi)$. The foliation we consider is the horizontal foliation given by $\varphi = \textrm{constant}$. The associated foliation Lie algebroid $\Pi A$ has adapted coordinates $(\vartheta, \varphi, \rmd \vartheta)$.  The homological vector field is given by  $\rmd_\mathcal{F} = \rmd \vartheta \partial_\vartheta$. As any leaf-wise (pseudo-)differential form has the structure $\alpha= \alpha_0(\vartheta, \varphi) + \rmd \vartheta \alpha_1(\vartheta, \varphi)$, we observe that  $\alpha \star \beta = \rmd \vartheta \partial_\vartheta \alpha_0 (\vartheta, \varphi)\, \beta_0(\vartheta, \varphi)$ as there are no non-zero leaf-wise $2$-forms.  We remark this also implies that $\alpha \star \beta \star \gamma =0$.  We will fix an orientation.\par 
Let $\tau$ be a derived cyclic $0$-cocycle. We have the decomposition $\tau(\alpha) = \tau_0(\alpha_0)+ \tau_1(\rmd \vartheta \alpha_1)$. Furthermore, we know that $\tau \circ \rmd_\mathcal{F} =0$.
\begin{itemize}
\item $\tau_0$ is unconstrained as there are no leaf-wise $-1$-forms, i.e., functions are never $\rmd_\mathcal{F}$-exact.
\item $\tau_1$ must annihilate $\vartheta$-derivatives. Via the standard theory of distributions on $S^1$ (which is obviously compact), we know that any functional that vanishes on exact one-forms must be a multiple of the integral over $S^1$. Thus, 
$$\tau_1(\rmd \vartheta \alpha_1) = \int_{T^2} \rmd \vartheta \, \rmd \varphi \, \mu(\varphi) \, \alpha_1(\vartheta, \varphi)\,,$$
where $\mu(\varphi)$ is a transverse invariant measure (via \cite{Ruelle:1975}). That is, $\tau_1$ is a classical  Ruelle--Sullivan current. 
\end{itemize}
To construct a derived cyclic $1$-cocycle, we set $\tau_0 =0$ and fix some transverse invariant measure $\mu(\varphi)$. That is, we will set $\tau$ to be a classical  Ruelle--Sullivan current.  We then select a vector field tangent to the leaves $X \in \Vect_{\mathcal{F}}(T^2)$. The leaf-wise Lie derivative is defined as $\mathcal{L}^{\mathcal{F}}_X :=  [\rmd_\mathcal{F}, i_X]$. We write $X = f(\vartheta, \varphi)\partial_\vartheta$.  Note that for $\gamma := \rmd \vartheta \, \gamma_1(\vartheta, \varphi)$, we have $\mathcal{L}^\mathcal{F}_X \gamma = \rmd \vartheta \partial_\vartheta (f \gamma_1)$. Then, for $\alpha = \alpha_0 + \rmd \vartheta \alpha_1$ we see that 
$$\tau \big( \mathcal{L}^\mathcal{F}_X \alpha \big)= \int_{T^2} \rmd \vartheta \rmd \varphi\, \mu(\varphi) \, \partial_\vartheta(f \alpha_1)\,.$$
Via the Fundamental Theorem of Calculus applied to the closed loop $S^1$, $\tau \big( \mathcal{L}^\mathcal{F}_X \alpha \big)=0$. Thus, $\tau$ is invariant under the leaf-wise Lie derivative.  We then define 
$$\phi^1_X(\alpha, \beta) := \tau\big(\alpha \star \mathcal{L}^{\mathcal{F}}_X \beta\big)= \tau\big(\alpha_0 \star X(\beta_0)\big) =  \int_{T^2} \rmd \vartheta \rmd \varphi\, \mu(\varphi) \,f(\vartheta, \varphi) \big(\partial_\vartheta(\alpha_0)\big) \big(\partial_\vartheta(\beta_0)\big) \,.$$
As the only non-zero components of the functional are related to $\alpha_0$ and $\beta_0$, we will consider $\phi^1_X(\alpha_0, \beta_0)$. Note that using the shifted grading $\overline{\alpha}_0 = \overline{\beta}_0 =1$. Via inspection, we see that $\phi^1_X(\alpha_0, \beta_0) = \phi^1_X(\beta_0, \alpha_0)$, and thus $\phi^1_X$ is shifted cyclic.\par 
Checking the action of the Hochschild coboundary operator, we have by definition
$$(b \phi^1_X)(\alpha, \beta, \gamma ) := \phi^1_X(\alpha \star \beta, \gamma) - \phi^1_X(\alpha, \beta \star \gamma)+(-1)^{\overline{\gamma}\, (\overline{\alpha} + \overline{\beta})}\, \phi^1_X(\gamma \star \alpha, \beta)\,.$$
By construction, $\phi^1_X$  is only non-zero if both entries contain zero-forms (functions). However, $\alpha \star \beta$, $\beta \star \gamma$  and $\gamma \star \alpha$ are one-forms. Thus, $b\phi^1_X =0$, and  so $\phi^1_X$ is a derived cyclic $1$-cocycle.
\end{example}
Example \ref{exa:T2} suggests a method of constructing derived cyclic $n$-cocycles for more general regular foliations provided the modular class of the foliation Lie algebroid vanishes. Suppose $\dim N = d$ and the leaves of the foliation are of dimension $p$. We set $k:= d-p$. We will take $N$ compact to avoid having to employ differential forms with compact support. We assume that $N$ is orientable and an orientation has been chosen.  For clarity, let us employ adapted coordinates $(x^a, y^\alpha, \rmd x^b)$ on $\Pi A \subset \Pi \sT N$. We select a classical Ruelle--Sullivan current, i.e., we define (locally)
\begin{equation}
\tau(\alpha) := \int_N \rmd^px\, \rmd^ky ~ \mu(y)\, \alpha_p(x,y)\,.
\end{equation}
As we take $N$ to be compact, $\tau\big(\mathcal{L}^\mathcal{F}_X \alpha \big)=0$, where $\mathcal{L}^\mathcal{F}_X := [\rmd_\mathcal{F}, i_X]$, where $X \in \Vect_\mathcal{F}(N)$. Suppose that for any $n \geq 1$ we have a collection of vector fields tangent to the leaves $\boldsymbol{X}:= \{X_1, X_2, \cdots , X_n  \}$ that pairwise commute, i.e. $[X_i, X_j]=0$. We then  define 
\begin{equation}\label{eqn:DerCycNCoc}
\phi^n_{\boldsymbol{X}}(\alpha_0, \alpha_1, \cdots , \alpha_n) := \sum_{\sigma\in S_n} \mathrm{sgn}(\sigma)\, \tau \big( \alpha_0 \star \mathcal{L}^\mathcal{F}_{X_{\sigma(1)}}\alpha_1 \star \cdots \star \mathcal{L}^\mathcal{F}_{X_{\sigma(n)}}\alpha_n \big)\,.
\end{equation}
Note that the Lie derivatives (as vector fields on $\Pi A$) are (shifted) even.  We observe that the above is an application of Connes' construction of a cyclic $n$-cocycle from an invariant trace and commuting derivations (see \cite{Connes:1985}). Thus, $\phi^n_{\boldsymbol{X}}$ defines a derived cyclic $n$-cocycle. \par   
We remark that Connes applied cyclic cohomology to the noncommutative convolution algebra constructed from the holonomy groupoid of a foliation (see \cite{Connes:1994}).  The derived cyclic cohomology is quite different and built from functionals of leaf-wise pseudo-differential forms.
\begin{example}\label{exa:TorusBund}
Let $N$ be a compact manifold with a free and proper action of the Lie group $T^p$. By the Quotient Manifold Theorem, the orbit space $N\slash T^p$ is a smooth manifold and $N \rightarrow N \slash T^p$ is a smooth submersion; we of course obtain a principal bundle. In particular, we have a regular foliation whose leaves are diffeomorphic to $T^p$.  It is a standard result that the Reeb class of such foliations vanishes - the foliation can be equipped with an auxiliary invariant metric rendering the foliation a Riemannian foliation. The Lie algebra of $T^p$ is  $\mathfrak{t} \cong\R^p$. By choosing an ordered basis compatible with the integral lattice $\{e_1, \cdots , e_p \}$, we construct a global frame of fundamental vector fields   $\boldsymbol{X} := \{ X_1, \cdots , X_p\}$ on $N$ using the standard method. These vector fields are tangent to the foliation and all pairwise commute, i.e, $\mathfrak{t}$ is abelian. For $1\leq n < p$, we define all the possible subsets $\boldsymbol{X}_{n, i} \subset \boldsymbol{X}$ consisting of $n$ elements.  We then define $\phi^n_{\boldsymbol{X}_{n,i}}$ via \eqref{eqn:DerCycNCoc}. For $n=p$, we define $\phi^p_{\boldsymbol{X}}$. Any change of ordered basis for $\mathfrak{t}$ is implemented on the fundamental vector fields via $g \in \mathsf{GL}(p, \Z)$. Note $\det(g) = \pm 1$. Thus, under change of the ordered basis we have $\phi^p_{\boldsymbol{X}} =  \det(g)\,\phi^p_{\boldsymbol{X}'} = \pm \phi^p_{\boldsymbol{X}} $. Thus, the derived cyclic $p$-cocycle is canonical up to an overall sign (the orientation of $T^p$).  From a starting basis, we have chosen an orientation of $T^p$, and we may consider changes of bases that respect the orientation and so $\det(g) = +1$. With this restriction, the derived cyclic $p$-cocycle is canonical (up to the choice of classical current).  
\end{example}
Example \ref{exa:TorusBund} generalises as follows. Suppose we have a free and proper action of a compact Lie group $G$ on a compact manifold $N$. As standard, we have a principal $G$-bundle. The modular class of the associated foliation Lie algebroid vanishes; one can equip the foliation with an auxiliary invariant metric which can be used to construct a transverse invariant volume form and, in turn, a  classical Ruelle--Sullivan current. The Cartan subalgebra $\mathfrak{t}\subset \mathfrak{g}$ (the largest abelian Lie algebra within $\mathfrak{g}$) generates a set of $r$ pairwise commuting vector fields via a choice of basis. Here $r$ is the rank of $G$.  We can then build derived cyclic $n$-cocycles up to degree $n = r$.  The $n = r$ derived cyclic cocycle is canonical up to a non-zero scalar multiple determined by the choice of ordered basis of the Cartan subalgebra, i.e., up to the determinant of the change of ordered basis, and a choice of the classical current.
%
%
\subsection{Applications to Deformations: Infinitesimal Deformations of Q-manifolds}
Let $(M,Q)$ be a Q-manifold and let $X \in \Vect(M)$ be an odd (infinitesimal) symmetry, i.e., an odd vector field that satisfies $\mathcal{L}_X Q = \mathcal{L}_Q X= 0$.  It is well known that the deformation problem for homological vector fields is controlled by the differential Lie algebra $(\Vect(M), [-,-], \mathcal{L}_Q)$. That is, we can infinitesimally deform $Q$ using an odd (infinitesimal) symmetry, i.e., using an odd vector field that is $\mathcal{L}_Q$-closed. We will return to this shortly. \par 
Suppose we select a (non-zero) $\tau \in HC^0_\lambda(M,Q)$. From this specified pair, we construct a derived cyclic $1$-cochain
\begin{equation}\label{eqn:1CoChain}
\phi(f_0, f_1) :=  \tau \big( (-1)^{\widetilde{f}_0} \, X(f_0 f_1)\big)\in C^1_\lambda(M,Q)\,.
\end{equation}
A quick calculation, which we omit, shows that $\phi(f_0, f_1) = - (-1)^{\overline{f}_0 \overline{f}_1}\, \phi(f_1, f_0)$, and so $\phi$ is indeed shifted cyclic. 
\begin{lemma}\label{lem:1CoCycle}
Let $(M,Q)$ be a Q-manifold and $X \in \Vect(M)$ be an odd infinitesimal symmetry. The derived cyclic $1$-cochain $\phi(f_0, f_1) :=  \tau \big( (-1)^{\widetilde{f}_0} \, X(f_0 f_1)\big)$ is a derived cyclic $1$-cocycle.
\end{lemma}
\begin{proof}
From the definition of the Hochschild coboundry operator \eqref{eqn:HochCobOp}, we have  
\begin{align*}
 b \phi(f_0,f_1, f_2 ) &= \phi(f_0 \star f_1, f_2) - \phi(f_0, f_1 \star f_2) + (-1)^{(\widetilde{f}_2 +1)(\widetilde{f}_0 + \widetilde{f}_1)}\, \phi(f_2 \star f_0, f_1)\\
& = \tau \big( (-1)^{\widetilde{f}_1 + 1} \, X(Q(f_0)f_1 f_2) + (-1)^{\widetilde{f_0} + \widetilde{f}_1+1}\,X(f_0 Q(f_1)f_2) \\ 
& + (-1)^{(\widetilde{f}_2 +1)(\widetilde{f}_0 + \widetilde{f}_1) + \widetilde{f}_0 + 1 }\,  X (Q(f_2) f_0 f_1) \big)\\
& = \tau \big( (-1)^{\widetilde{f}_1 + 1} \, X(Q(f_0)f_1 f_2) + (-1)^{\widetilde{f_0} + \widetilde{f}_1+1}\,X(f_0 Q(f_1)f_2) \\ 
& + (-1)^{\widetilde{f}_0 + 1 }\,  X ( f_0 f_1 Q(f_2)) \big)\\
&= \tau \big( (-1)^{\widetilde{f}_1+1}\, X (Q(f_0 f_1 f_2))\big)\\
\intertext{Then using $[Q, X] =0$ and Lemma \ref{lem:0Cocycles}, i.e., $\tau \circ Q =0$, we obtain}
& = \tau \big( Q((-1)^{\widetilde{f}_1}\, X(f_0 f_1 f_2) )\big) =0.
\end{align*} 
\end{proof}
Using an odd infinitesimal symmetry, we can infinitesimally deform $Q$ as
\begin{equation}
Q_{(t)} :=  Q + t \, X\,,
\end{equation}
where $\widetilde{t} =0$ and we drop terms $\cO(t^2)$ and higher. Directly we see that
$$Q_{(t)}^2  = (Q + t \, X)^2 =  Q^2 + t \big( Q X + X Q \big) + \cO(t^2) = t\, [Q, X] + \cO(t^2) = t\, \mathcal{L}_Q X +\cO(t^2)\,.$$
Thus, up to first-order in $t$,  $Q^2_{(t)}=0$. The derived product is similarly deformed 
$$f\star_{(t)} g = f \star g + (-1)^{\widetilde{f}}\, t \, X(f)g\,,$$
and leads to a deformed Hochschild coboundary operator which we denote as $b_{(t)}$ (which is well-defined dropping terms $\cO(t^2)$ and higher). \par   
Suppose we wish to look for an infinitesimally deformed derived cyclic $0$-cocycle $\tau_{(t)} := \tau + t \, \sigma$, with $\sigma \in C^0_\lambda(M,Q)$. By definition, we have that $b_{(t)}\tau_{(t)}=0$ (to first-order).
\begin{lemma}\label{lem:bClosed}
Let $(M,Q)$ be a Q-manifold, and $X \in \Vect(M)$ be an odd infinitesimal symmetry. Assuming $\tau_{(t)} := \tau + t \, \sigma$ is an infinitesimally deformed derived cyclic $0$-cocycle, then $\phi$ is $b$-exact.
\end{lemma}
\begin{proof}
Expanding $b_{(t)}\tau_{(t)}=0$ and retaining only the zeroth and first-order in $t$ terms we have
\begin{align*}
b_{(t)}\tau_{(t)}(f_0, f_1) &= b\tau(f_0, f_1) + t \, \big( b\sigma(f_0, f_1) + \tau((-1)^{\widetilde{f}_0} \, X(f_0 f_1)  )  \big)\\
\intertext{Then using $b \tau =0$ and the definition of $\phi$ \eqref{eqn:1CoChain}, we have}
&= t \, \big( b\sigma(f_0, f_1) + \phi(f_0 ,f_1)  )  \big) =0\,.
\end{align*} 
Thus, as $f_0$ and $f_1$ are arbitrary, we have $\phi = - b \sigma$, and the lemma is established.
\end{proof}
\begin{theorem}\label{the:ObsDeform}
Let $(M,Q)$ be a Q-manifold, $X \in \Vect(M)$ be an odd infinitesimal symmetry, and $\tau$ be a derived cyclic $0$-cocycle. Then an infinitesimally deformed derived cyclic $0$-cocycle $\tau_{(t)} = \tau + t\, \sigma $  exists if and only if the derived cyclic class $[\phi]$ is zero.
\end{theorem}
\begin{proof}
Lemma \ref{lem:1CoCycle} established that $\phi$ is a derived cyclic $1$-cocycle, and so defines a derived cyclic class $[\phi] \in HC^1_\lambda(M,Q)$. Lemma \ref{lem:bClosed} establishes that if $\tau_{(t)}$ exists, then $\phi = - b \sigma$, which directly implies that the relevant derived cyclic class is zero. Conversely,  if  $\phi = - b \sigma$ for some $\sigma \in C^0_\lambda(M,Q)$, i.e, $[\phi]$ is zero, then we may define   $\tau_{(t)} = \tau + t\, \sigma$ such that $b_{(t)}\tau_{(t)}=0$. 
\end{proof}
\begin{definition}
The derived cyclic class  $[\phi] \in HC^1_\lambda(M,Q)$ defined by the derived cyclic $1$-cocycle $\phi$ (see \eqref{eqn:1CoChain}) is referred to as the \emph{deformation class} of $\tau$ along $X$.
\end{definition}
%
%
\section{Concluding Remarks}\label{sec:ConRem}
We examined the notion of a derived product in the setting of Q-manifolds.  As this multiplication is Grassmann odd, we require a shift in the grading to obtain what we referred to as a derived associative algebra.  This derived product is inherently noncommutative.  Alongside other results and constructions, we applied cyclic cohomology to the derived associative algebra.  Note that this cohomology theory is distinct from the standard cohomology of a Q-manifold. \par 
In standard cyclic cohomology applied to a noncommutative algebra, cyclic $n$-cocycles are analogues of $n$-dimensional de Rham currents, i.e., linear functionals on differential forms. Derived cyclic $0$-cocycles are understood as BRST-invariant (shifted) traces. Derived higher cyclic $n$-cocycles are interpreted as BRST-invariant currents, i,e., multi-linear functionals on functions that vanish on BRST-exact functions.  It was shown that the infinitesimal deformation problem for derived cyclic $0$-cocycles is controlled by the derived cyclic cohomology.  A little more carefully, classes in $HC^1_\lambda(M,Q)$ classify the obstruction to finding infinitesimally deformed cyclic $0$-cocycles. \par 
\smallskip 
\noindent \textbf{Future Directions:}\
\begin{enumerate}[i)]
\item Explicit examples of derived cyclic $n$-cocycles ($n\geq 1$) and derived cyclic cohomology groups require construction/calculation. Including (multi-)weights, for example ghost and antifield number, etc., may simplify the situation and give clearer links with physics and gauge mechanics. Moreover, the application of derived cyclic cohomology in the theory of (regular) foliations awaits proper exploration. However, as the derived associative algebra is non-unital, we realise that the standard methods/tools of calculating Connes' cyclic cohomology cannot be directly applied (once the grading and shifted cyclic properties are taken into account). For example, the typical methods of Connes' SBI sequence, mixed complexes using Connes' periodicity operator, Morita theory, and application of the K\"unneth theorems, require the algebras to be unital. Thus, we either have to work with less standard  frameworks that can cope with non-unital algebras, or adjoin a unit element (see Subsection \ref{subsec:AdjUniEle}). It must still be stressed that (standard) Hochschild cohomology and cyclic cohomology for noncommutative algebras are, in general, very difficult to calculate. In particular, it is common to have infinite-dimensional cochain complexes that grow exponentially in size, which make direct computations largely intractable. 
\item We cannot directly use K\"unneth-type decompositions to simplify the calculation of the derived cyclic cohomology of a Q-manifold; for the case of cyclic cohomology of superalgebras see Kassel \cite{Kassel:1986}. For one, the standard decompositions require the algebras to be unital, though this can be circumvented by the unitalisation as stated above.  However, there is a more fundamental reason.  Let us set $A = \big ( C^\infty(M_1), \star _1\big)$ and  $B = \big ( C^\infty(M_2), \star _2\big)$ simply considered associative algebras.  Then in the standard  K\"unneth-type formula we encounter  $A \otimes B$ (we ignore completions of the tensor product for this discussion), where the algebra is defined component-wise as
$$(a_1 \otimes b_1) \star (a_2 \otimes b_2) = \pm (a_1 \star_1 a_2)\otimes (b_1 \star_2 b_2)\,. $$
The issue is that this component-wise multiplication is structurally distinct from the derived product on $C^\infty(M_1 \times M_2)$ generated by $Q= Q_1 + Q_2$.  In particular, 
$$(f_1 \otimes f_2) \star (g_1 \otimes g_2) =  (f_1 \star_1 g_1) \otimes (f_2 g_2) \pm (f_1 g_1)\otimes (f_2 \star_2 g_2)\,,$$
as $Q$ acts as a derivation over the tensor product. This implies that ``off-the-shelf'' decompositions of the cohomology cannot be directly applied to derived cyclic cohomology. Thus, developing a modified K\"unneth-type isomorphism  for derived cyclic cohomology is highly desirable to aid computations, assuming such an isomorphism exists.  Developing the formalism in the context of differential graded algebras would allow the subtle issues with functional analysis to be avoided.
\item Formally, the notions and constructions of this work generalise to infinite-dimensional Q-manifolds such as the Fr\'{e}chet or convenient supermanifolds of fields and antifields encountered in the BV--BRST framework. A rigorous treatment of the derived cyclic cohomology will require a restriction to local functionals to be applied to field theory. Developing  `local derived cyclic cohomology' may yield new algebraic invariants for classical gauge field theories and their space of observables.
\end{enumerate}
In conclusion, derived cyclic cohomology of Q-manifolds ties to central ideas in differential geometry and mathematical physics. Given how prevalent Q-manifolds are, it is expected that further applications than those started here will be uncovered.
%
%
\section*{Acknowledgements}  
The author cordially thanks Janusz Grabowski and Steven Duplij for comments on earlier drafts of this work.

%
%

\end{document}